\documentclass[12pt,a4paper]{amsart}

\usepackage{comment}
\usepackage{amsmath,amsfonts,amssymb,extarrows}
\usepackage{color}
\usepackage[hmargin=2cm,vmargin=2cm]{geometry}

\usepackage{hyperref}
\usepackage{nameref,zref-xr}                    
\usepackage[all]{xy}           

\newcommand{\ttheta}{\widetilde{\theta}}

\newcommand{\mbZ}{\mathbb Z}
\newcommand{\mbC}{\mathbb C}

\newcommand{\oM}{\overline{\mathcal M}}

\newcommand{\tu}{{\widetilde u}}
\newcommand{\og}{\overline g}

\newcommand{\of}{\overline f}
\newcommand{\oY}{\overline Y}
\newcommand{\oX}{\overline X}
\newcommand{\hLambda}{\widehat\Lambda}
\newcommand{\cT}{\mathcal T}
\def\cM{{\mathcal{M}}}
\def\oM{{\overline{\mathcal{M}}}}
\def\CP{{{\mathbb C}{\mathbb P}}}
\renewcommand{\Im}{\rm Im}

\def\d{{\partial}}

\newcommand{\eps}{\varepsilon}

\newcommand{\cA}{\mathcal A}
\newcommand{\hcA}{\widehat{\mathcal A}}
\newcommand{\DR}{\mathrm{DR}}

\newcommand{\even}{\mathrm{even}}
\newcommand{\ct}{\mathrm{ct}}

\newcommand{\Coef}{\mathrm{Coef}}

\newcommand{\Deg}{\mathrm{Deg}}

\newcommand{\gl}{\mathrm{gl}}

\newcommand{\tR}{\widetilde{R}}
\newcommand{\hu}{\widehat{u}}

\newcommand{\tP}{\widetilde{P}}

\newcommand{\mcO}{\mathcal{O}}

\def\g{\mathfrak{g}}

\def\d{\partial}

\def\f{\frac}

\newcommand{\beq}{\begin{equation}}
\newcommand{\eeq}{\end{equation}}
\newcommand{\orig}{\mathrm{orig}}
\newcommand{\Id}{\mathrm{Id}}

\newcommand{\Mat}{\mathrm{Mat}}
\newcommand{\End}{\mathrm{End}}

\newcommand{\oF}{\overline{F}}

\newcommand{\ot}{{\overline{t}}}

\def\un{{1\!\! 1}}

\newcommand{\tPsi}{\widetilde{\Psi}}
\newcommand{\diag}{\mathrm{diag}}
\newcommand{\tD}{\widetilde{D}}
\newcommand{\tGamma}{\widetilde{\Gamma}}

\newcommand{\triv}{\mathrm{triv}}

\newcommand{\Lie}{\mathrm{Lie}}
\newcommand{\otau}{\overline{\tau}}
\newcommand{\Sh}{\mathrm{Sh}}
\newcommand{\oH}{\overline{H}}
\newcommand{\tY}{\widetilde{Y}}
\newcommand{\desc}{\mathrm{desc}}

\newcommand{\cX}{\mathcal{X}}
\newcommand{\br}{\overline{r}}
\newcommand{\Ker}{\mathrm{Ker}}
\newcommand{\Hom}{\mathrm{Hom}}
\newcommand{\hE}{\widehat{E}}
\newcommand{\ogamma}{\overline{\gamma}}
\newcommand{\tL}{\widetilde{L}}
\newcommand{\hcX}{\widehat{\mathcal{X}}}
\newcommand{\hS}{\widehat{S}}

\newtheorem{theorem}{Theorem}[section]
\newtheorem{proposition}[theorem]{Proposition}
\newtheorem{lemma}[theorem]{Lemma}

\newtheorem{conjecture}[theorem]{Conjecture}

\theoremstyle{definition}

\newtheorem{definition}[theorem]{Definition}

\newtheorem{remark}[theorem]{Remark}

\numberwithin{equation}{section}

\begin{document}

\title{Flat F-manifolds, F-CohFTs, and integrable hierarchies}

\author{Alessandro Arsie}
\address{A.~Arsie:\newline Department of Mathematics and Statistics, The University of Toledo, 2801W,\newline 
Bancroft St., 43606 Toledo, OH, USA}
\email{alessandro.arsie@utoledo.edu}

\author{Alexandr Buryak}
\address{A.~Buryak:\newline Faculty of Mathematics, National Research University Higher School of Economics, \newline
6 Usacheva str., Moscow, 119048, Russian Federation; \newline
Center for Advanced Studies, Skolkovo Institute of Science and Technology,\newline
1 Nobel str., Moscow, 143026, Russian Federation; \newline
P.G. Demidov Yaroslavl State University,\newline 
14 Sovetskaya str., Yaroslavl, 150003, Russian Federation}
\email{aburyak@hse.ru}

\author{Paolo Lorenzoni}
\address{P.~Lorenzoni:\newline Dipartimento di Matematica e Applicazioni, Universit\`a di Milano-Bicocca, \newline
Via Roberto Cozzi 53, I-20125 Milano, Italy and INFN sezione di Milano-Bicocca}
\email{paolo.lorenzoni@unimib.it}

\author{Paolo Rossi}
\address{P.~Rossi:\newline Dipartimento di Matematica ``Tullio Levi-Civita'', Universit\`a degli Studi di Padova,\newline
Via Trieste 63, 35121 Padova, Italy}
\email{paolo.rossi@math.unipd.it}

\begin{abstract}
We define the double ramification hierarchy associated to an F-cohomological field theory and use this construction to prove that the principal hierarchy of any semisimple (homogeneous) flat F-manifold possesses a (homogeneous) integrable dispersive deformation at all orders in the dispersion parameter. The proof is based on the reconstruction of an F-CohFT starting from a semisimple flat F-manifold and additional data in genus $1$, obtained in our previous work. Our construction of these dispersive deformations is quite explicit and we compute several examples. In particular, we provide a complete classification of rank $1$ hierarchies of DR type at the order $9$ approximation in the dispersion parameter and of homogeneous DR hierarchies associated with all $2$-dimensional homogeneous flat F-manifolds at genus $1$ approximation.
\end{abstract}

\date{\today}

\maketitle

\tableofcontents

\section*{Introduction}
Since Witten's conjecture~\cite{Wit91} and its proof by Kontsevich~\cite{Kon92}, there have been growing and fruitful interactions between the area of integrable hierarchies of PDEs and algebraic geometry of the moduli spaces of algebraic curves. In this context, and in connection with topological field theory, Dubrovin introduced in the 90s the notion of \emph{Frobenius manifold}~\cite{Dub96}, a differential-geometric structure that encodes genus-zero information af a cohomological field theory (CohFT) on the moduli space of stable curves, besides having far reaching connections with other areas of mathematics.\\ 

From the point of view of integrable systems, given a Dubrovin--Frobenius manifold, there exists an associated integrable hierarchy of Hamiltonian quasilinear PDEs called \emph{Dubrovin's principal hierarchy}, or simply \emph{principal hierarchy}. An important problem in the theory of integrable systems consists in constructing a full dispersive hierarchy starting from its dispersionless limit.\\

In the framework of moduli spaces, the principal hierarchy associated to a Dubrovin--Frobenius manifold and its dispersive deformation should satisfy additional constraints coming from the intersection theory of the CohFT. In the semisimple case, there exist two different (but conjecturally Miura-equivalent \cite{Bur15,BDGR18,BGR19}) constructions defining such dispersive deformations: 
\begin{enumerate}
\item The Dubrovin--Zhang construction \cite{DZ01} is based on the idea that the partition function of the corresponding CohFT in all genera is the logarithm of the tau-function of a special solution (called the \emph{topological solution}) to a full dispersive hierarchy (the {\it DZ hierarchy}). One can construct the hierarchy itself starting from this tau-function, and it turns out that the the principal hierarchy is the dispersionless limit of DZ hierarchy. Moreover the full DZ hierarchy and the principal hierarchy are related by a special change of dependent variables, called a \emph{quasi-Miura transformation}, which can be uniquely determined in the semisimple case from genus zero information.\\
\item The double ramification construction, introduced by one of the authors in \cite{Bur15}, is based on the definition of an infinite set of commuting Hamiltonian densities \cite{BR16a} in terms of intersection numbers of the CohFT, the double ramification cycles and other natural tatutological classes on the moduli space of curves.
\end{enumerate}
For both constructions  and in the (homogeneous) semisimple case, the reconstruction of the full dispersive hierarchy from its dispersionless limit (the principal hierarchy of the Dubrovin--Frobenius manifold encoding the genus $0$ part of the CohFT) is possible thanks to the Givental--Telemann reconstruction theorem for the CohFT itself from its genus $0$ part \cite{Tel12, Giv01}.\\
 
Notice that, by construction, the dispersionless limits of both the DZ and DR hierarchies coincide with the principal hierarchy of the Dubrovin--Frobenius manifold underlying the CohFT.\\

In the last 20 years, it has been observed that many constructions related to Dubrovin--Frobenius manifolds can be extended to a more general setting (\cite{Sab98,Get04,Man05,LPR09,SZ11,AL13,Lor14,KMS,AL17,DH17,BR18,KMS18,AL19,BB19,ABLR20}). For instance, it was observed in \cite{LPR09} that the notion of principal hierarchy does not require the existence of an invariant flat metric. This leads naturally to the consideration of the generalization of Dubrovin--Frobenius manifolds, called F-manifolds with compatible flat structure \cite{Man05} or simply \emph{flat F-manifolds} \cite{LPR09}, obtained by replacing a flat metric with a flat torsionless connection and keeping all the axioms of Dubrovin--Frobenius manifolds apart from those involving explicitly the metric and not just the associated Levi--Civita connection. In flat coordinates for the flat connection, the flows of the principal hierarchy are systems of conservation laws. In the case of Dubrovin--Frobenius manifolds, the presence of an invariant flat metric has to deal with the presence of a local Hamiltonian structure.\\
    
In this paper we construct (homogeneous) double ramification hierarchies starting from a (homogeneous) CohFT. In particular, in the semisimple case, leveraging on the results of \cite{ABLR20}, this provides dispersive deformations of the principal hierarchy associated to a semisimple (homogeneous) flat F-manifold. The existence of these dispersive integrable deformations relies on:
\begin{enumerate}
\item a generalization of the notion of cohomological field theory, called \emph{F-cohomological field theory} (or F-CohFT for short) introduced in \cite{BR18,ABLR20};
\item a reconstruction theorem for a semisimple (homogeneous) F-CohFT starting from a flat F-manifold and additional data in genus $1$ \cite{ABLR20};
\item the definition of an infinite set of commuting flows (the {\it DR hierarchy}) in terms of intersection numbers of the F-CohFT, the double ramification cycles, the top Hodge class, and psi classes on the moduli space of stable curves.
\end{enumerate}  

\bigskip

The paper is organized as follows.\\

Section~\ref{section:DR hierarchy of F-CohFT} is devoted to the construction of the DR hierarchy of an F-CohFT (see also~\cite{BR18}). The main properties of this hierarchy are given in terms of densities of local vector fields on the formal loop space and a special basis for their integrals of motion. We also consider the additional properties of the hierarchy in the case of a homogeneous F-CohFT.\\ 

%More precisely, if the underlying flat F-manifold is homogeneous, then the associated family of F-CohFTs contains a subfamily of homogeneous F-CohFTs whose associated integrable hierarchies are also homogeneous. However, unlike the case of Dubrovin--Frobenius manifolds and CohFTs, these homogeneous F-CohFTs can have different conformal dimensions and, as a consequence, there exist homogeneous deformations of the principal hierarchy with different degrees of homogeneity.

In Section~\ref{section:principal hierarchy and dispersive deformations}, after recalling the definition of a flat F-manifold and the construction of its associated principal hierarchy, we present our main result: given an arbitrary semisimple flat F-manifold and an associated principal hierarchy, we construct a family of dispersive integrable deformations of the principal hierarchy. These deformations, called the {\it descendant DR hierarchies}, come from the family of DR hierarchies associated to a family of F-CohFTs parameterized by a semisimple point of our flat F-manifold. The descendant DR hierarchy depends on a choice of a certain vector field on the flat F-manifold, which we call a {\it framing}. We prove that the descendant DR hierarchies corresponding to different framings are not related to each other by a Miura transformation that is close to identity.\\

In Section~\ref{section:classification}, we discuss the role of (descendant) DR hierarchies in the problem of classification of integrable deformations of integrable dispersionless systems of conservation laws. One can impose various constraints for such integrable deformations, and we discuss the corresponding results (mostly at the approximation up to some finite power of $\eps$) for flat F-manifolds of dimension $1$ and $2$ in Section~\ref{biflatdef} and~\ref{section:classification of rank 1 hierarchies of DR type}. In Section~\ref{sec:final}, we briefly mention the problem of computing general integrable deformations of principal hierarchies of flat F-manifolds. It was conjectured in \cite{AL18} that the equivalence classes of such deformations are labeled by certain functional parameters called {\it Miura invariants}. In the case of Dubrovin--Frobenius manifolds and bihamiltonian deformations, these invariants are equivalent to central invariants, which are known to classify deformations of semisimple local bihamiltonian structures of hydrodynamic type (\cite{DLZ06,CPS18}).\\

\subsection*{Acknowledgements}
The work of A.~B. is supported by the Russian Science Foundation (Grant no. 20-71-10110). P.~L. is supported by MIUR - FFABR funds 2017 and by funds of H2020-MSCA-RISE-2017 Project No. 778010 IPaDEGAN.\\

%%%%%%%%%%%%%%%%%%%%%%%%%%%%%%%%%%%%%%%%%%%%%%%%%%%%%%%%%%%%
%%%%%%%%%%%%%%%%%%%%%%%%%%%%%%%%%%%%%%%%%%%%%%%%%%%%%%%%%%%%

\section{Double ramification hierarchy of an F-CohFT}\label{section:DR hierarchy of F-CohFT}

In this section, we associate to any F-CohFT with a vector space~$V$ an infinite sequence of commuting vector fields on the formal loop space of $V$, i.e., an infinite sequence of compatible systems of evolutionary PDEs of rank $N:=\dim V$ (in particular, in the form of conservation laws). This construction is a generalization of the double ramification hierarchy of \cite{Bur15,BR16a} to the context of F-CohFTs and enjoys most of its properties (for instance, recursion formulas for the higher symmetries), but loses in general the Hamiltonian nature.\\

\subsection{F-cohomological field theories}

We recall from \cite{BR18,ABLR20} the definition of an F-cohomological field theory on the moduli space $\oM_{g,n}$ of stable curves of genus $g$ with $n$ marked points. We will denote by $H^*(X)$ the cohomology ring with coefficients in $\mbC$ of a topological space~$X$. When considering the moduli space of stable curves, $X=\oM_{g,n}$, the even part $H^\even(\oM_{g,n})$ in the cohomology ring $H^*(\oM_{g,n})$ can optionally be replaced by the Chow ring $A^*(\oM_{g,n})$.

\medskip

\begin{definition}\label{definition:F-CohFT}
An \emph{F-cohomological field theory} (or F-CohFT) is a system of linear maps 
$$
c_{g,n+1}\colon V^*\otimes V^{\otimes n} \to H^\even(\oM_{g,n+1}),\qquad 2g-1+n>0,
$$
where $V$ is an arbitrary finite dimensional vector space, together with a special element $e\in V$, called the \emph{unit}, such that, chosen any basis $e_1,\ldots,e_{\dim V}$ of $V$ and the dual basis $e^1,\ldots,e^{\dim V}$ of $V^*$, the following axioms are satisfied:
\begin{itemize}
\item[(i)] The maps $c_{g,n+1}$ are equivariant with respect to the $S_n$-action permuting the $n$ copies of~$V$ in $V^*\otimes V^{\otimes n}$ and the last $n$ marked points in $\oM_{g,n+1}$, respectively.
\item[(ii)] $\pi^* c_{g,n+1}(e^{\alpha_0}\otimes \otimes_{i=1}^n e_{\alpha_i}) = c_{g,n+2}(e^{\alpha_0}\otimes \otimes_{i=1}^n  e_{\alpha_i}\otimes e)$ for $1 \leq\alpha_0,\alpha_1,\ldots,\alpha_n\leq \dim V$, where $\pi\colon\oM_{g,n+2}\to\oM_{g,n+1}$ is the map that forgets the last marked point.\\
Moreover, $c_{0,3}(e^{\alpha}\otimes e_\beta \otimes e) = \delta^\alpha_\beta$ for $1\leq \alpha,\beta\leq \dim V$.
\item[(iii)] $\gl^* c_{g_1+g_2,n_1+n_2+1}(e^{\alpha_0}\otimes \otimes_{i=1}^{n_1+n_2} e_{\alpha_i}) = c_{g_1,n_1+2}(e^{\alpha_0}\otimes \otimes_{i\in I} e_{\alpha_i} \otimes e_\mu)\otimes c_{g_2,n_2+1}(e^{\mu}\otimes \otimes_{j\in J} e_{\alpha_j})$ for $1 \leq\alpha_0,\alpha_1,\ldots,\alpha_{n_1+n_2}\leq \dim V$, where $I \sqcup J = \{2,\ldots,n_1+n_2+1\}$, $|I|=n_1$, $|J|=n_2$, and $\gl\colon\oM_{g_1,n_1+2}\times\oM_{g_2,n_2+1}\to \oM_{g_1+g_2,n_1+n_2+1}$ is the corresponding gluing map.
\end{itemize}
\end{definition}

\medskip

There is an obvious generalization of the notion of an F-CohFT where the maps $c_{g,n+1}$ take value in $H^\even(\oM_{g,n+1})\otimes K$, where $K$ is a $\mbC$-algebra. We will call such objects {\it F-cohomological field theories with coefficients in $K$}.

\medskip

\begin{definition}
An F-CohFT $c_{g,n+1}\colon V^*\otimes V^{\otimes n}\to H^\even(\oM_{g,n+1})$ is called \emph{homogeneous} if there exists an operator $Q\in\End(V)$, a vector $\br\in V$, and a complex constant~$\gamma$ such that $Q e=0$ and the following condition is satisfied:
\begin{multline}\label{eq:homogeneous F-CohFT}
\Deg\circ c_{g,n+1}+\pi_{*}\circ c_{g,n+2}\circ(\otimes\br)=\\
=c_{g,n+1}\circ\left(-Q^t\otimes\Id^{\otimes n}+\sum_{i+j=n-1}\Id\otimes\Id^{\otimes i}\otimes Q\otimes \Id^{\otimes j}\right)+\gamma g c_{g,n+1},
\end{multline}
where $\Deg\in\End(H^*(\oM_{g,n}))$ is the operator acting on $H^i(\oM_{g,n})$ by the multiplication by~$\frac{i}{2}$, $\pi\colon\oM_{g,n+2}\to\oM_{g,n+1}$ is the map that forgets the last marked point, $\otimes\br\colon V^*\otimes V^{\otimes n}\to V^*\otimes V^{\otimes(n+1)}$ is the operator of tensor multiplication from the right by $\br$, and $Q^t\in\End(V^*)$ is the transposed operator. The constant $\gamma$ is called the \emph{conformal dimension} of our F-CohFT.
\end{definition}

\smallskip

\begin{remark}\label{remark:homogeneous F-CohFTs}
Our definition of a homogeneous F-CohFT is slightly more general, than the one from the paper~\cite{ABLR20} where the operator $Q$ was required to be diagonalizable. However, it is easy to see that all the results from~\cite{ABLR20} about homogeneous F-CohFTs are true with the new definition (see also Section~\ref{subsection:flat F-manifolds} with a new definition of a homogeneous flat F-manifold). An example of a homogeneous F-CohFT with a nondiagonalizable operator~$Q$ will appear in Section~\ref{section:classification}.
\end{remark}

\medskip

\subsection{Vector fields on the the formal loop space}\label{subsection:formal loop space}

Let $\hcA$ and $\hLambda$ be the spaces of differential polynomials and local functionals in formal (even) variables $u^\alpha_k$, $1\leq\alpha\leq N$, $k\geq 0$, and $\eps$, with the differential grading $\deg_{\d_x} u^\alpha_k = k$, $\deg_{\d_x} \eps = -1$, where the definitions and the notations are taken from \cite[Section 2.1]{Ros17}.\\

The \emph{space of densities of local multivector fields} (on the formal loop space of $V$) is the supercommutative associative algebra
$$
\hcA^{\bullet}:=\mbC[[u^*,\theta_*]][u^*_{>0},\theta_{*,>0}][[\eps]],
$$
where the new formal variables $\theta_{\alpha,k}$, $1\leq\alpha\leq N$, $k\geq 0$, are odd (anti-commuting among themselves and commuting with $\eps$ and $u^\alpha_k$) with $\deg_{\d_x} \theta_{\alpha,k}:=k$, $u^\alpha:=u^\alpha_0$, and $\theta_{\alpha}:=\theta_{\alpha,0}$, and the symbol~$*$, as an index, denotes any of the allowed values for that index. The algebra~$\hcA^{\bullet}$ is endowed with the {\it super grading}, denoted by $\deg_\theta$, which is defined by $\deg_\theta\theta_{\alpha,k}:=1$ and $\deg_\theta u^\alpha_k=\deg_\theta\eps:=0$. The sub-vector space of $\hcA^\bullet$ homogeneous of super degree $i\geq 0$ is denoted by $\hcA^i$ and called the \emph{space of densities of local $i$-vector fields}. We have $\hcA = \hcA^0$, while~$\hcA^1$ is called the \emph{space of densities of local vector fields}. The homogeneous component of the space~$\hcA^i$ of differential degree $k$ will be denoted by $(\hcA^i)^{[k]}$.\\

The operator $\d_x$ is extended from $\hcA$ to $\hcA^{\bullet}$ as the super-derivation
$$
\d_x: = \sum_{k\geq 0} \left(u^\alpha_{k+1} \frac{\d }{\d u^\alpha_k}+\theta_{\alpha,k+1} \frac{\d }{\d \theta_{\alpha,k}} \right),
$$
where, here and in what follows, we perform summation over repeated Greek indices. The \emph{space of local multivector fields} is defined as 
$$
\hLambda^\bullet:=\hcA^\bullet/(\Im\,\d_x\oplus\mbC[[\eps]])
$$
and, for $i\geq 0$, the \emph{space of local $i$-vector fields} $\hLambda^i$ is the image of $\hcA^i$ in the quotient. If $f\in \hcA^\bullet$, its image in $\hLambda^\bullet$ is denoted by $\overline{f} = \int f dx$. As before, $\hLambda = \hLambda^0$, and $\hLambda^1$ is called the \emph{space of local vector fields}. Naturally, the spaces $\hLambda^i$ inherit the differential grading $\deg_{\d_x}$.\\

For any $1 \leq \alpha\leq N$, we define the \emph{(super) variational derivatives}
$$
\frac{\delta}{\delta u^\alpha}:=\sum_{k\geq 0} (-\d_x)^k\frac{\d}{\d u^\alpha_k}\ , \qquad \frac{\delta}{\delta \theta_\alpha}:=\sum_{k\geq 0} (-\d_x)^k\frac{\d}{\d \theta_{\alpha,k}},
$$ 
which are well defined on $\hLambda^\bullet$ since they vanish on $\Im\,\d_x \oplus \mbC[[\eps]]$.\\

The \emph{Schouten--Nijenhuis bracket} $[\cdot,\cdot]\colon\hLambda^i \times \hcA^j \to \hcA^{i+j-1}$ is defined by
\begin{equation}\label{eq:SN bracket lift 1}
[\of,g]:=\sum_{k\geq 0}\left( \d_x^k\left(\frac{\delta\of}{\delta \theta_\alpha}\right) \frac{\d g}{\d u^{\alpha}_k}+ (-1)^i \d_x^k\left(\frac{\delta\of}{\delta u^\alpha}\right) \frac{\d g}{\d \theta_{\alpha,k}} \right).
\end{equation}
This Schouten--Nijenhuis bracket is a lift of the Schouten--Nijenhuis bracket $[\cdot,\cdot]\colon \hLambda^i \times \hLambda^j \to \hLambda^{i+j-1}$ defined by
\begin{equation}\label{eq:SN bracket}
[\of,\og]:=\int \left(\frac{\delta \of}{\delta \theta_\alpha} \frac{\delta \og}{\delta u^\alpha}+ (-1)^i \frac{\delta \of}{\delta u^\alpha} \frac{\delta \og}{\delta \theta_\alpha}\right) dx.
\end{equation}
A further lift of the Schouten--Nijenhuis bracket to $ \hcA^i \times \hcA^j$ can be defined employing formal Dirac delta functions, similarly to what was done in \cite{BR16b} for the quantum commutator of two differential polynomials,
\begin{equation}\label{eq:SN bracket lift 2}
[f(x),g(y)]:=\sum_{k,l\geq 0}\left(\frac{\d f}{\d \theta_{\alpha,k}}(x)\, \frac{\d g}{\d u^\alpha_l}(y)\,  \d_x^k \d_y^l \delta(x-y)+(-1)^i \frac{\d f}{\d u^\alpha_k}(x)\, \frac{\d g}{\d \theta_{\alpha,l}}(y)\, \d_x^k \d_y^l \delta(x-y)\right).
\end{equation}
Taking the integral with respect to $x$ of formula \eqref{eq:SN bracket lift 2}, using that $\int \delta(x-y) g(y) dx = g(y)$, reproduces indeed formula \eqref{eq:SN bracket lift 1}, and further integration with respect to~$y$ gives \eqref{eq:SN bracket}.\\

As usual, for $i=j=1$, the above Schouten--Nijenhuis brackets are called the \emph{Lie brackets}. For $i=1$ and $j=0$, the Schouten--Nijenhuis brackets reduce simply to the differentiation of (a density of) a local functional along (a density of) a vector field, from which we see that the symbol~$\theta_{\alpha,k}$ can be interpreted as the operator $\d_x^k \circ \frac{\delta}{\delta u^\alpha}: \hLambda\to\hcA$.\\

Given a local vector field $\oX \in \hLambda^1$, there is a unique representative $X \in \hcA^1$ of $\oX$ such that $X = X^\alpha \theta_\alpha$ with $X^\alpha \in \hcA$. This representative is given by $X=\frac{\delta \oX}{\delta \theta_\alpha} \theta_\alpha$. The system of evolutionary PDEs associated to $\oX$ is
\begin{equation}\label{eq:evolutionary PDEs}
\frac{\d u^\alpha}{\d t} = \frac{\delta \oX}{\delta \theta_\alpha}(u^*_*;\eps), \qquad \alpha=1, \ldots,N.
\end{equation}
Two systems of evolutionary PDEs
\begin{align*}
&\frac{\d u^\alpha}{\d t} = \frac{\delta \oX}{\delta \theta_\alpha}(u^*_*;\eps), \qquad \alpha=1, \ldots,N,\\
&\frac{\d u^\alpha}{\d s} = \frac{\delta \oY}{\delta \theta_\alpha}(u^*_*;\eps), \qquad \alpha=1, \ldots,N,
\end{align*}
are compatible, in the sense that, for any $1\leq \alpha\leq N$, $\frac{\d}{\d t} \frac{\d u^\alpha}{\d s} = \frac{\d}{\d s}\frac{\d u^\alpha}{ \d t}$, if and only if the associated local vector fields $\oX,\oY \in \hLambda^1$ satisfy $[\oX,\oY] = 0$.\\

Under a Miura transformation (see \cite[Section 2.1]{Ros17} for more details) of the form 
\begin{align}
&\tu^\alpha=\tu^\alpha(u^*_*;\eps) \in \hcA^{[0]}=(\hcA^0)^{[0]},\qquad 1\leq\alpha\leq N,\label{eq:Miura1}\\
&\tu^*|_{u^*_*=0}=0,\qquad \left.\det\left(\frac{\d\tu^*}{\d u^*}\right)\right|_{u^*_*=0}\ne 0,\label{eq:Miura2}
\end{align}
the generators  $u^*_*$ and $\theta_{*,*}$ of $\hcA^\bullet$ transform according to the formulae
$$
u^\alpha_k = \d_x^k u^\alpha(\tu^*_*;\eps), \qquad \theta_{\alpha,k} = \d_x^k \left(\sum_{s\geq 0} (-\d_x)^s \left(\left.\frac{\d \tu^\mu}{\d u^\alpha_s}\right|_{u^*_*=u^*_*(\tu^*_*;\eps)} \ttheta_{\mu}\right)\right),\qquad 1\leq \alpha\leq N,\quad k\geq 0,
$$
where $u^\alpha(\tu^*_*;\eps)$ is obtained by inverting $\tu^\alpha = \tu^\alpha(u^*_*;\eps)$ order by order in $\eps$. For a local vector field, these formulae give
$$
\oX = \int(X^\alpha\theta_\alpha)dx = \int\left(\left.\left(\sum_{s\geq 0}\frac{\d \tu^\alpha}{\d u^\mu_s} \d_x^s X^\mu\right)\right|_{u^*_*=u^*_*(\tu^*_*;\eps)}\ttheta_\alpha\right) dx,
$$
from which we obtain that a system of evolutionary PDEs \eqref{eq:evolutionary PDEs} transforms into
$$
\frac{\d \tu^\alpha}{\d t} = \widetilde{X}^\alpha(\tu^*_*;\eps) = \left.\left(\sum_{s\geq 0}\frac{\d \tu^\alpha}{\d u^\mu_s} \d_x^s X^\mu\right)\right|_{u^*_*=u^*_*(\tu^*_*;\eps)}, \qquad \alpha=1, \ldots,N.
$$

\bigskip

Performing the change of formal variables
\begin{gather}\label{eq:u-p change}
u^\alpha_k = \d_x^k \left(\sum_{a\in\mbZ} p_a^\alpha e^{iax}\right), \qquad \theta_{\alpha,k} = \d_x^k \left(\sum_{a\in \mbZ} q_{\alpha,a} e^{iax} \right), \qquad 1\leq\alpha\leq N,\quad k\geq 0,
\end{gather}
one can rewrite a density of a local multivector field $f(u^*_*,\theta_{*,*};\eps)\in(\hcA^m)^{[d]}$ as a formal Fourier series
$$
f=\sum_{\substack{n,s\geq 0 \\ a_1,\ldots,a_n \in \mbZ\\b_1,\ldots,b_m \in \mbZ}} f^{a_1,\ldots,a_n,b_1,\ldots,b_m}_{\alpha_1,\ldots,\alpha_n,\beta_1,\ldots,\beta_m;s}\, \eps^s\,  p^{\alpha_1}_{a_1}\ldots p^{\alpha_n}_{a_n}\, q_{\beta_1,b_1}\ldots q_{\beta_m,b_m}\, e^{i\left(\sum_{j=1}^n a_j + \sum_{j=1}^m b_j\right)x}, 
$$
where the coefficient $f^{a_1,\ldots,a_n,b_1,\ldots,b_m}_{\alpha_1,\ldots,\alpha_n,\beta_1,\ldots,\beta_m;s}$, as a function of the indices $a_1,\ldots,a_n,b_1,\ldots,b_m$, is a homogeneous polynomial of degree $s+d$. Formal Fourier series of this type form a supercommutative associative algebra where the formal variables $q_{*,*}$ are odd. Moreover, the local multivector field $\of$ corresponds to the constant term of the Fourier series. Similarly to the variables $\theta_{*,*}$, one should interpret the variable $q_{\alpha,a}$ to represent the vector $\frac{\d}{\d p^{\alpha}_{-a}}$. This is coherent with the following formulae for the variational derivatives in the variables $p^*_*$ and $q_{*,*}$:
$$
\frac{\delta}{\delta u^\alpha} = \sum_{a\in \mbZ} e^{iax} \frac{\d}{\d p^\alpha_{-a}},\qquad \frac{\delta}{\delta \theta_\alpha} = \sum_{a\in \mbZ} e^{iax} \frac{\d}{\d q_{\alpha,-a}},
$$
acting on local multivector fields to give densities of local multivector fields. Accordingly, using the formal Fourier expansion $\delta(x) = \sum_{a\in \mbZ} e^{iax}$ for the formal Dirac delta function, it is easy to obtain the formula for the Schouten--Nijenhuis bracket \eqref{eq:SN bracket lift 2} on $\hcA^i \times \hcA^j$ in the new variables:
\begin{equation}\label{eq:SN bracket in p variables}
[f(x),g(y)] = \sum_{a\in \mbZ} \left( \frac{\d f}{\d q_{\alpha,a}}(x)\, \frac{\d g}{\d p^\alpha_{-a}}(y)\, + (-1)^i \frac{\d f}{\d p^\alpha_{-a}}(x)\, \frac{\d g}{\d q_{\alpha,a}}(y)\,\right),
\end{equation}
from which analogues of \eqref{eq:SN bracket lift 1} and \eqref{eq:SN bracket} are easily obtained by integration in $x$ and then $y$.\\

\subsection{Densities of local vector fields for the DR hierarchy}\label{subsection:construction of DR hierarchy}

Denote by $\psi_i\in H^2(\oM_{g,n})$ the $i$-th \emph{psi class}, which is the first Chern class of the line bundle over~$\oM_{g,n}$ formed by the cotangent lines at the $i$-th marked point. Denote by~$\mathbb E$ the rank~$g$ Hodge vector bundle over~$\oM_{g,n}$ whose fibers are the spaces of holomorphic one-forms on stable curves. Let $\lambda_j:=c_j(\mathbb E)\in H^{2j}(\oM_{g,n})$, these classes are called the \emph{Hodge classes}.\\ 

For any $a_1,\dots,a_n\in \mbZ$, $\sum_{i=1}^n a_i =0$, denote by $\DR_g(a_1,\ldots,a_n) \in H^{2g}(\oM_{g,n})$ the {\it double ramification (DR) cycle}. We refer the reader, for example, to~\cite{BSSZ15} for the definition of the DR cycle on~$\oM_{g,n}$, which is based on the notion of a stable map to $\CP^1$ relative to~$0$ and~$\infty$. If not all the multiplicities $a_i$ are equal to zero, then one can think of the class $\DR_g(a_1,\ldots,a_n)$ as the Poincar\'e dual to a compactification in~$\oM_{g,n}$ of the locus of pointed smooth curves~$(C;p_1,\ldots,p_n)$ satisfying $\mathcal O_C\left(\sum_{i=1}^n a_ip_i\right)\cong\mathcal O_C$. Consider the Poincar\'e dual to the double ramification cycle~$\DR_g(a_1,\ldots,a_n)$ in the space $\oM_{g,n}$. It is an element of $H_{2(2g-3+n)}(\oM_{g,n})$, and abusing notation it is also denoted by $\DR_g(a_1,\ldots,a_n)$.\\

The restriction $\DR_g(a_1,\ldots,a_n)\big|_{\cM_{g,n}^{\ct}}$, where $\cM_{g,n}^{\ct}$ is the moduli space of stable curves of compact type, is a homogeneous polynomial in $a_1,\ldots,a_n$ of degree $2g$ with the coefficients in~$H^{2g}(\cM_{g,n}^{\ct})$. This follows from Hain's formula~\cite{Hai13} for the version of the DR cycle defined using the universal Jacobian over~$\cM^\ct_{g,n}$ and the result of the paper~\cite{MW13}, where it is proved that the two versions of the DR cycle coincide on $\cM^\ct_{g,n}$ (the polynomiality of the DR cycle on~$\oM_{g,n}$ is proved in~\cite{JPPZ17}). The polynomiality of the DR cycle on $\cM^\ct_{g,n}$ together with the fact that $\lambda_g$ vanishes on~$\oM_{g,n}\setminus\cM_{g,n}^{\ct}$ (see, e.g.,~\cite[Section~0.4]{FP00}) imply that the cohomology class $\lambda_g\DR_g(-\sum_{j=1}^n a_j,a_1,\ldots,a_n) \in H^{4g}(\oM_{g,n+1})$ is a degree $2g$ homogeneous polynomial in the coefficients $a_1,\ldots,a_n$.\\

Given a vector space $V$ with $\dim V = N$ and a basis $e_1,\ldots,e_N\in V$, let $c_{g,n+1}\colon V^*\otimes V^{\otimes n} \to H^\even(\oM_{g,n+1})$ be an F-CohFT with unit $e=A^\mu e_\mu $. For $1\leq\beta\leq N$ and $d\geq 0$, we define the following system of formal Fourier series:
\begin{equation}\label{eq:DR densities in p variables}
Y_{\beta,d} := -\hspace{-0.2cm}\sum _{\substack{g,n\geq 0,\,2g+n>0\\a,a_1,\ldots,a_n \in \mbZ}}\hspace{-0.2cm}\frac{i a (-\eps^2)^g}{n!}\left(\int_{\DR_g(a,-a-\sum_{j=1}^n a_j,a_1,\ldots,a_n)} \hspace{-3.6cm}\lambda_g \psi_2^d c_{g,n+2}(e^\alpha\otimes e_\beta\otimes \otimes_{j=1}^n e_{\alpha_j})\right)q_{\alpha,a} \left(\prod_{j=1}^n p^{\alpha_j}_{a_j}\right)e^{i\left(a+\sum_{j=1}^n a_j\right)x},
\end{equation}
which, thanks to the polynomiality property of the DR cycle, can be rewritten as a system of densities of local vector fields $Y_{\beta,d}\in(\hcA^1)^{[1]}$ as
\begin{equation}\label{eq:DR densities}
\begin{split}
Y_{\beta,d}=-\sum_{\substack{g,n\geq 0,\,2g+n>0\\k,k_1,\ldots,k_n\geq 0\\k+\sum_{j=1}^nk_j=2g}}\frac{\eps^{2g}}{n!} \Coef_{a^k(a_1)^{k_1}\ldots(a_n)^{k_n}}&\left(\int_{\DR_g(a,-a-\sum_{j=1}^n a_j,a_1,\ldots,a_n)}\hspace{-3.6cm}\lambda_g \psi_2^d c_{g,n+2}(e^\alpha\otimes e_\beta\otimes \otimes_{j=1}^n e_{\alpha_j}) 
\right)\theta_{\alpha,k+1}\ \prod_{j=1}^n u^{\alpha_j}_{k_j}.
\end{split}
\end{equation}
To this definition, we add the extra densities $Y_{\beta,-1}:=- \theta_{\beta,1}$, $1\leq \beta\leq N$.\\

The {\it double ramification hierarchy} associated to the given F-CohFT is the infinite system of local vector fields $\oY_{\beta,d}$, $1\leq \beta\leq N$, $d\geq-1$, associated with the above densities or, in terms of evolutionary PDEs, the system
\begin{gather}\label{eq:DR hierarchy}
\frac{\d u^\alpha}{\d t^\beta_d}=\d_x P^\alpha_{\beta,d},\qquad 1\le \alpha,\beta\le N,\quad d\ge 0,
\end{gather}
where
\begin{equation}\label{eq:P-polynomials for DR hierarchy}
P^\alpha_{\beta,d}:=\sum_{\substack{g,n\geq 0,\,2g+n>0\\k_1,\ldots,k_n\geq 0\\\sum_{j=1}^n k_j=2g}} \frac{\eps^{2g}}{n!} \Coef_{(a_1)^{k_1}\ldots(a_n)^{k_n}} \left(\int_{\DR_g(-\sum_{j=1}^n a_j,0,a_1,\ldots,a_n)} \hspace{-2.3cm}\lambda_g \psi_2^d c_{g,n+2}(e^\alpha\otimes e_\beta\otimes \otimes_{j=1}^n e_{\alpha_j}) \right)\prod_{j=1}^n u^{\alpha_j}_{k_j}.
\end{equation}
Let us adopt the convention $P^\alpha_{\beta,-1}:=\delta^\alpha_\beta$. Notice that the system of evolutionary PDEs \eqref{eq:DR hierarchy} carries strictly less information than the corresponding densities \eqref{eq:DR densities}.  We have the following result.

\medskip

\begin{theorem}[\cite{BR18}]
All the equations of the DR hierarchy (\ref{eq:DR hierarchy}) are compatible with each other, namely, 
$$
\frac{\d}{\d t^{\beta_2}_{d_2}}\left(\frac{\d u^\alpha}{\d t^{\beta_1}_{d_1}}\right) = \frac{\d}{\d t^{\beta_1}_{d_1}}\left(\frac{\d u^\alpha}{\d t^{\beta_2}_{d_2}}\right),\qquad 1\leq \alpha,\beta_1,\beta_2\leq N,\quad d_1,d_2\geq 0.
$$
\end{theorem}

\medskip

This theorem is proved in \cite{BR18}, but we give another proof in Theorem \ref{theorem:Main} (see part~(ii)). For $1\leq\beta_1,\beta_2\leq N$ and $d_1,d_2\geq 0$, let us define the generating series
\begin{equation}\label{eq:DR 2-point vector}
\begin{split}
Y_{\beta_1,d_1;\beta_2,d_2}(x,y):=-\hspace{-0.35cm}\sum _{\substack{g,n\geq 0\\a,b_1,b_2,a_1,\ldots,a_n \in \mbZ}}\hspace{-0.35cm}\frac{ia(-\eps^2)^g}{n!} &\left(\int_{\DR_g(a,b_1,b_2,a_1,\ldots,a_n)} \hspace{-1.6cm}\lambda_g \psi_2^{d_1} \psi_3^{d_2} c_{g,n+3}(e^\alpha\otimes e_{\beta_1}\otimes e_{\beta_2}\otimes\otimes_{j=1}^n e_{\alpha_j})\right)\cdot\\
&\cdot q_{\alpha,a} \left(\prod_{j=1}^n p^{\alpha_j}_{a_j}\right)e^{-ib_1x}\,e^{-ib_2y},
\end{split}
\end{equation}
where we adopt the convention that $\DR_g(a,b_1,b_2,a_1,\ldots,a_n):=0$ when $a+b_1+b_2+\sum_{j=1}^n a_j\ne 0$. To this definition, for future convenience, we add $Y_{\beta_1,-1;\beta_2,d}(x,y)=Y_{\beta_1,d;\beta_2,-1}(x,y):=0$, $1\leq \beta_1,\beta_2\leq N$, $d\geq 0$.\\

We will use the symbol~$\un$, as an index, to denote the sum over the values $1\le\alpha\le N$ for that index with the coefficients $A^\alpha$. For example, $Y_{\un,d}:=A^\mu Y_{\mu,d}$, $\theta_{\un,k} := A^\mu \theta_{\mu,k}$, and $\frac{\d}{\d t^\un_d}:=A^\mu\frac{\d}{\d t^\mu_d}$.

\medskip

\begin{theorem}\label{theorem:Main}
For all $1\leq \beta_1,\beta_2\leq N$ and $d_1,d_2\ge -1$ such that $d_1+d_2\geq -1$, we have
\begin{itemize}
\item[(i)] $\displaystyle [Y_{\beta_2,d_2}(y),Y_{\beta_1,d_1}(x)] = \d_x Y_{\beta_1,d_1+1;\beta_2,d_2}(x,y) - \d_y Y_{\beta_1,d_1;\beta_2,d_2+1}(x,y)$;
\item[(ii)] $\displaystyle[\oY_{\beta_2,d_2},\oY_{\beta_1,d_1}] =0$;
\item[(iii)] $\displaystyle[\oY_{\un,1},Y_{\beta_1,d_1}] = \d_x (D-1) Y_{\beta_1,d_1+1}$, where $\displaystyle D:= \sum_{k\geq 0} \left(u^\alpha_k \frac{\d}{\d u^\alpha_k} + \theta_{\alpha,k}\frac{\d}{\d \theta_{\alpha,k}}\right) + \eps \frac{\d}{\d \eps}$;
\item[(iv)] $\displaystyle[\oY_{\beta_2,0},Y_{\beta_1,d_1}] = \d_x \frac{\d }{\d u^{\beta_2}}Y_{\beta_1,d_1+1}$;
\item[(v)] $\displaystyle Y_{\un,0} =- u^\alpha \theta_{\alpha,1} + \d_x^2 S$,  $S\in (\hcA^1)^{[-1]}$, which implies $\displaystyle \frac{\d u^\alpha}{\d t^\un_0} = \d_x u^\alpha$ for $1\leq\alpha\leq N$;
\item[(vi)] $\displaystyle \frac{\d}{\d u^\un} Y_{\beta_1,d_1+1} = Y_{\beta_1,d_1}$, $\displaystyle \frac{\d}{\d u^\un} P^{\beta_2}_{\beta_1,d_1+1} = P^{\beta_2}_{\beta_1,d_1}$.
\item[(vii)] $\displaystyle\frac{\d}{\d u^{\beta_2}} P^{\beta_1}_{\un,1} = D P^{\beta_1}_{\beta_2,0}$.
\end{itemize}
\end{theorem}
\begin{proof}
For $n\ge 0$, let us use the notation $[n]$ for the set $\{1,\ldots,n\}$.\\

Let us prove part~(i). If $d_1=-1$ or $d_2=-1$, then the statement easily follows from the definitions. For $d_1,d_2\ge 0$, the statement is analogous to \cite[Lemma 3.3]{BR16b}, and we use \cite[Corollary 2.2]{BSSZ15}, describing the intersection of the psi classes with the DR cycle, together with the fact that that $\lambda_g$ vanishes on~$\oM_{g,n}\setminus\cM_{g,n}^{\ct}$. Let $n\ge 0$ and consider integers $a_1,\ldots,a_{n+3}$ with the vanishing sum. For a subset $I=\{i_1,\ldots,i_{|I|}\}\subset [n+3]$, $i_1<i_2<\ldots<i_{|I|}$, denote by~$A_I$ the string $a_{i_1},a_{i_2},\ldots,i_{|I|}$. For $I,J\subset[n+3]\setminus\{2,3\}$ with $I\sqcup J=[n+3]\setminus\{2,3\}$, and for $g_1,g_2>0$ with $2g_1+|I|>0$, $2g_2+|J|>0$, let us denote by $\DR_{g_1}(a_2,A_I,-k)\boxtimes \DR_{g_2}(a_3,A_J,k)$ the cycle in~$\oM_{g_1+g_2,n+3}$ obtained by gluing the two DR cycles at the marked points labeled by the integers~$-k$ and~$k$, respectively. Here, the coefficient $a_j$, $1\le j\le n+3$, is attached to the marked point~$j$. Then we have
\begin{equation}\label{eq:psi-psi times DR}
(a_2\psi_2 - a_3\psi_3)\lambda_g \DR_g(A_{[n+3]})=\sum_{\substack{I\sqcup J=[n+3]\setminus\{2,3\}\\ k\in \mbZ,\, g_1\ge 0,\,g_2 \ge 0\\g_1+g_2=g\\2g_1+|I|,\,2g_2+|J|>0}} \lambda_g\cdot k\cdot\DR_{g_1}(a_2,A_I,-k)\boxtimes\DR_{g_2}(a_3,A_J,k).
\end{equation}

\bigskip

One then needs to intersect this relation with the class $-a_1 e^{-i a_2 x} e^{-i a_3 y}\psi_2^{d_1} \psi_3^{d_2} c_{g,n+3}(e^{\alpha_1}\otimes \otimes_{i=2}^{n+3}e_{\alpha_i})$, where, as usual, the covector $e^{\alpha_1}$ is attached to the marked point $1$ and each vector~$e_{\alpha_i}$ is attached to the marked point $i$. Thanks to the gluing axiom of the F-CohFT, by the definitions~\eqref{eq:DR densities in p variables} and~\eqref{eq:DR 2-point vector}, and after setting $\alpha_2=\beta_1$ and $\alpha_3=\beta_2$, the left-hand side of~equation~\eqref{eq:psi-psi times DR} produces the right-hand side of the equation in part~(i) and depending on whether, in the above sum, the marked point~$1$ belongs to the subset $I$ or $J$, we obtain either of the two terms in the Lie bracket on the left-hand side of the equation in part~(i).\\

Part~(ii) is immediately obtained from (i) upon integration in both $x$ and $y$.\\

Part~(iii) is obtained from (i) after setting $\beta_2=\un$, $d_2=1$ and integrating in $y$. The generating series $\int Y_{\beta_1,d_1+1;\un,1} (x,y)dy$ reduces to $(D-1) Y_{\beta_1,d_1+1}$ thanks to the following simple equality:
$$
\int_{\DR_g(a,b_1,0,a_1,\ldots,a_n)} \hspace{-2.4cm}\lambda_g \psi_2^{d_1+1}\psi_3 c_{g,n+3}(e^\alpha\otimes e_{\beta_1}\otimes e\otimes \otimes_{j=1}^n e_{\alpha_j}) =(2g+n) \int_{\DR_g(a,b_1,a_1,\ldots,a_n)} \hspace{-2.2cm}\lambda_g\psi_2^{d_1+1} c_{g,n+2}(e^\alpha\otimes e_{\beta_1}\otimes\otimes_{j=1}^n e_{\alpha_j}),
$$
which is in turn a consequence of the following behavior of the involved cohomology classes with respect to the morphism $\pi\colon\oM_{g,n+3}\to\oM_{g,n+2}$ forgetting the third marked point:
\begin{align}
&\DR_g(a,b_1,0,a_1,\ldots,a_n)=\pi^*\DR_g(a,b_1,a_1,\ldots,a_n),\label{eq:DR and pullback of DR}\\
&c_{g,n+3}(e^\alpha\otimes e_{\beta_1}\otimes e\otimes\otimes_{j=1}^n e_{\alpha_j})=\pi^*c_{g,n+2}(e^\alpha\otimes e_{\beta_1}\otimes\otimes_{j=1}^n e_{\alpha_j}),\label{eq:c and pullback of c}\\
&\lambda_g = \pi^* \lambda_g,\qquad\pi_*(\psi_2^{d_1+1}\psi_3) = (2g+n) \psi_2^{d_1+1}.\label{eq:lambda and pullback of lambda}
\end{align}
Indeed, the operator $D$ multiplies each term of $Y_{\beta_1,d_1+1}$ by the number of variables~$\eps$,~$u^*_*$, and~$\theta_{*,*}$ appearing in that term, i.e., by $2g+n+1$.\\

Part~(iv) is similarly obtained from (i) by setting $d_2=0$ and integrating in $y$, as $\int  Y_{\beta_1,d_1+1;\beta_2,0}(x,y) dy$ reduces by definition to $\frac{\d}{\d u^{\beta_2}}Y_{\beta_1,d_1+1}$.\\

To deduce (v), we consider formula \eqref{eq:DR densities in p variables} and notice that, for $(g,n)\neq(0,1)$,
$$
\int_{\DR_g(a,-a-\sum_{j=1}^n a_j,a_1,\ldots,a_n)} \hspace{-2.2cm}\lambda_g c_{g,n+2}(e^\alpha\otimes e\otimes\otimes_{j=1}^n e_{\alpha_j}) = \int_{\pi_*(\lambda_g\DR_g(a,-a-\sum_{j=1}^n a_j,a_1,\ldots,a_n))} \hspace{-2.9cm}c_{g,n+1}(e^\alpha\otimes\otimes_{j=1}^n e_{\alpha_j}),
$$
and $\pi_*(\lambda_g\DR_g(a,-a-\sum_{j=1}^n a_j,a_1,\ldots,a_n))$ is divisible by $(a+\sum_{j=1}^n a_j)^2$ as proved in \cite[Lemma~5.1]{BDGR18}, where $\pi\colon\oM_{g,n+2}\to\oM_{g,n+1}$ is the map forgetting the second marked point. When $g=0$ and $n=1$, we have instead $\DR_0(a,-a-a_1,a_1)=1$, $\lambda_0=1$, and $c_{0,3}(e^\alpha\otimes e\otimes e_{\alpha_1})= \delta^{\alpha}_{\alpha_1}$, which gives the desired result.\\

Part~(vi) immediately follows from parts~(iv),~(v), the properties $\Ker\left(\d_x|_{\hcA^1}\right)=0$, $\Ker\left(\d_x|_{\hcA}\right)=\mbC[[\eps]]$, and the fact $\d_x P^{\beta_2}_{\beta_1,d_1}=\frac{\delta}{\delta\theta_{\beta_2}}\oY_{\beta_1,d_1}$.\\

For part~(vii), we compute $\d_x\frac{\d P^{\beta_1}_{\un,1}}{\d u^{\beta_2}}=\frac{\delta}{\delta\theta_{\beta_1}}\int Y_{\un,1;\beta_2,0}\,dx\,dy=\frac{\delta}{\delta\theta_{\beta_1}}(D-1) \oY_{\beta_2,0}=D \frac{\delta}{\delta \theta_{\beta_1}} \oY_{\beta_2,0}=D\d_x P^{\beta_1}_{\beta_2,0}=\d_x D P^{\beta_1}_{\beta_2,0}$.
\end{proof}

\medskip

\subsection{Densities of integrals of motion for the DR hierarchy} 

The DR hierarchy of a CohFT is a Hamiltonian integrable system \cite{Bur15, BR16a}, so the Hamiltonians both generate the commuting vector fields and provide integrals of motion for the hierarchy. In the non-Hamiltonian F-CohFT case, integrals of motion have a separate geometric definition in terms of intersection numbers on the moduli space of curves. For $1\leq \beta\leq N$ and $d\geq0$, we define the following system of formal Fourier series:
\begin{equation}\label{eq:DR densities iom in p variables}
\begin{split}
g^{\beta,d} := \sum _{\substack{g,n\geq 0\\2g+n-1>0\\a_1,\ldots,a_n \in\mbZ}}  \frac{(-\eps^2)^g}{n!} \left(\int_{\DR_g(-\sum_{j=1}^n a_j,a_1,\ldots,a_n)} \hspace{-2.3cm}\lambda_g \psi_1^d c_{g,n+1}(e^\beta\otimes \otimes_{j=1}^n e_{\alpha_j})\right) \left(\prod_{j=1}^n p^{\alpha_j}_{a_j}\right)e^{i\left(\sum_{j=1}^n a_j\right)x},
\end{split}
\end{equation}
which, thanks to the polynomiality property of the DR cycle, can be rewritten as differential polynomials $g^{\beta,d}\in\hcA^{[0]}$ as
\begin{equation}\label{eq:DR densities iom}
\begin{split}
g^{\beta,d}=\sum_{\substack{g,n\geq 0,\,2g+n-1>0\\k_1,\ldots,k_n\geq 0\\\sum_{j=1}^n k_j=2g}} \frac{\eps^{2g}}{n!} \Coef_{(a_1)^{k_1}\ldots(a_n)^{k_n}} \left(\int_{\DR_g(-\sum_{j=1}^n a_j,a_1,\ldots,a_n)} \hspace{-2.3cm}\lambda_g \psi_1^d c_{g,n+2}(e^\beta\otimes \otimes_{j=1}^n e_{\alpha_j}) 
\right) \prod_{j=1}^n u^{\alpha_j}_{k_j} .
\end{split}
\end{equation}
To this definition, we add the extra densities of conserved quantities $g^{\beta,-1}: = u^\beta $, $1\leq \beta\leq N$, and the ``primary'' local vector field $\oY := -\int g^{\beta,0} \theta_{\beta,1}\, dx$ or, in other words, $\d_x g^{\beta,0}=\frac{\delta \oY}{\delta \theta_\beta}$, $1\leq \beta\leq N$.\\

Finally, for $1\leq\beta_1,\beta_2\leq N$ and $d_1,d_2\geq 0$, let us define the generating series
\begin{equation}\label{eq:DR 2-point vector iom}
\begin{split}
g^{\beta_1,d_1}_{\beta_2,d_2}(x,y):=\hspace{-0.3cm}\sum_{\substack{g,n\geq 0,\,2g+n>0\\b_1,b_2,a_1,\ldots,a_n \in \mbZ}}\hspace{-0.3cm} \frac{(-\eps^2)^g}{n!} &\left(\int_{\DR_g(b_1,b_2,a_1,\ldots,a_n)} \hspace{-2.5cm}\lambda_g \psi_1^{d_1} \psi_2^{d_2} c_{g,n+2}(e^{\beta_1}\otimes e_{\beta_2}\otimes\otimes_{j=1}^n e_{\alpha_j})\right)  \left(\prod_{j=1}^n p^{\alpha_j}_{a_j}\right)e^{-ib_1x}\,e^{-ib_2y}.
\end{split}
\end{equation}
To this definition, for future convenience, we add $g^{\beta_1,-1}_{\beta_2,d}(x,y)=g^{\beta_1,d}_{\beta_2,-1}(x,y):=0$, $1\leq \beta_1,\beta_2\leq N$, $d\geq 0$. 

\medskip

\begin{theorem}\label{theorem:Main integrals}
For all $1\leq \beta_1,\beta_2\leq N$ and $d_1,d_2\ge -1$ such that $d_1+d_2\geq -1$, we have

\begin{itemize}
\item[(i)] $\displaystyle [Y_{\beta_1,d_1}(y),g^{\beta_2,d_2}(x)] = \d_x g^{\beta_2,d_2+1}_{\beta_1,d_1}(x,y)- \d_y g^{\beta_2,d_2}_{\beta_1,d_1+1}(x,y)$;

\item[(ii)] $\displaystyle [\oY_{\beta_1,d_1},\og^{\beta_2,d_2}]=0$;

\item[(iii)] $\displaystyle[\oY_{\un,1},g^{\beta_2,d_2}] = \d_x (D-1) g^{\beta_2,d_2+1}$;

\item[(iv)] $\displaystyle[\oY_{\beta_1,0},g^{\beta_2,d_2}] = \d_x \frac{\d }{\d u^{\beta_1}}g^{\beta_2,d_2+1}$;

\item[(v)] $\displaystyle\frac{\d}{\d u^\un} g^{\beta_1,d_1+1} = g^{\beta_1,d_1}$;

\item[(vi)] $\displaystyle \oY_{\un,1} = (D-2) \oY$;

\item[(vii)] $\displaystyle\oY_{\beta,0} = \frac{\d }{\d u^\beta}\oY$.
\end{itemize}
\end{theorem}
\begin{proof}
The proof of (i) is completely analogous to the proof of (i) in Theorem~\ref{theorem:Main}. For $d_1=-1$ or $d_2=-1$, the statement easily follows from the definitions. Suppose $d_1,d_2\ge 0$. Let $n\ge 0$ and consider integers  $a_1,\ldots,a_{n+2}$ with the vanishing sum. Let us write the same relation as~\eqref{eq:psi-psi times DR}, but with the psi classes taken at other marked points:
\begin{equation*}
(a_1\psi_1 - a_2 \psi_2)\lambda_g \DR_g(A_{[n+2]})=\sum_{\substack{I\sqcup J=[n+2]\setminus\{1,2\}\\ k\in \mbZ,\, g_1\ge 0,\,g_2 \ge 0\\g_1+g_2=g\\2g_1+|I|,\,2g_2+|J|>0}}\lambda_g\cdot k\cdot\DR_{g_1}(a_1,A_I,-k)\boxtimes\DR_{g_2}(a_2,A_J,k).
\end{equation*}
Intersecting this relation with the class $ (-i)e^{-a_1ix}e^{-a_2iy}\psi_1^{d_2}\psi_2^{d_1} c_{g,n+2}(e^{\alpha_1}\otimes \otimes_{j=2}^{n+2}e_{\alpha_j})$ and forming the corresponding generating series, we obtain part~(i) (after setting $\alpha_2=\beta_1$ and $\alpha_1=\beta_2$).\\

The proof of (ii) to (iv) follows strictly the arguments in the proof of the corresponding parts in Theorem~\ref{theorem:Main}.\\

The proof of part (v) is the same as the proof of part (vi) in Theorem~\ref{theorem:Main}.\\

For the proof of (vi), consider the equation of part (iii) with $d_2=-1$. Multiplying it by $\theta_{\beta_2}$, summing over $\beta_2$, and integrating over $x$ we obtain, on the left-hand side,
$$
\int[\oY_{\un,1},u^{\beta_2}]\theta_{\beta_2}dx=\int \frac{\delta \oY_{\un,1}}{\delta \theta_{\beta_2}}\theta_{\beta_2} dx = \oY_{\un,1}
$$
and, on the right-hand side,
$$
\int \d_x\left((D-1)g^{\beta_2,0}\right) \theta_{\beta_2} dx = -(D-2) \int g^{\beta_2,0} \theta_{\beta_2,1} dx = (D-2) \oY.
$$

\bigskip

Part~(vii) is proved in an analogous fashion starting from (iv).
\end{proof}

\medskip

\subsection{Homogeneous DR hierarchies}

Let $Y_{\beta,q} \in (\hcA^1)^{[1]}$ and $g^{\alpha,p}\in \hcA^{[0]}$, $1\leq \beta,\alpha\leq N$, $q,p\geq -1$, be the densities of local vector fields and of integrals of motion of the DR hierarchy associated to a homogeneous rank $N$ F-CohFT. Let 
$$
c^\alpha_{\beta\gamma}:=c_{0,3}(e^\alpha\otimes e_\beta\otimes e_\gamma)\in\mbC
$$
for $1\leq\alpha,\beta,\gamma\leq N$.\\

Consider the following vector field on the space of densities of local multivector fields on the formal loop space:
$$
\hE_\gamma:=\sum_{k\geq 0}\left(\left((\delta^\alpha_\beta-q^\alpha_\beta)u^\beta_k+\delta_{k,0} r^\alpha\right) \frac{\d}{\d u^\alpha_k}-(\delta^\alpha_\beta-q^\alpha_\beta)\theta_{\alpha,k}\frac{\d}{\d \theta_{\beta,k}}\right)+\frac{1-\gamma}{2} \eps \frac{\d}{\d \eps},
$$
where $q^\beta_\alpha e_\beta:=Qe_\alpha$ and $r^\alpha e_\alpha:=\overline{r}$. For convenience, let us define $Y_{\alpha,-2}=P^\beta_{\alpha,-2}:=0$ and $g^{\alpha,-2}:=A^\alpha$ for all $1\leq\alpha,\beta\leq N$.

\medskip

\begin{proposition}\label{proposition:homogeneous DR hierarchy}
For all $1\leq \alpha\leq N$ and $d\geq -1$, we have

\begin{itemize} 
\item[(i)] $\displaystyle \hE_\gamma(Y_{\alpha,d}) = d Y_{\alpha,d}+q^\beta_\alpha Y_{\beta,d}+r^\gamma c^\mu_{\gamma\alpha} Y_{\mu,d-1}$;

\item[(ii)] $\displaystyle \hE_\gamma(P^\alpha_{\beta,d}) = (d+1) P^\alpha_{\beta,d}+q^\gamma_\beta P^\alpha_{\gamma,d}-q^\alpha_\gamma P^\gamma_{\beta,d}+r^\gamma c^\mu_{\gamma\beta} P^\alpha_{\mu,d-1}$;

\item[(iii)] $\displaystyle \hE_\gamma(g^{\alpha,d})  = (d+2)g^{\alpha,d}-q^\alpha_\beta g^{\beta,d}+r^\gamma c^\alpha_{\gamma\mu} g^{\mu,d-1}$;

\item[(iv)] $\displaystyle \hE_\gamma(\oY) = \oY - r^\gamma c^\beta_{\gamma\mu} \int u^\mu \theta_{\beta,1} dx$.
\end{itemize}
\end{proposition}
\begin{proof}
The proof is a simple consequence of equation~\eqref{eq:homogeneous F-CohFT} together with dimension counting for the intersection numbers involved in the definitions of $g^{\alpha,d}$, $Y_{\alpha,d}$, and $\oY$ and the fact that $\pi^*\psi_i^d = \psi_i^d-\delta^0_{i,n+1}\pi^*\psi_i^{d-1}$, $1\le i\le n$, $d\ge 1$, where $\pi\colon\oM_{g,n+1}\to\oM_{g,n}$ forgets the last marked point and~$\delta^0_{i,n+1}$ is the closure in $\oM_{g,n+1}$ of the locus of stable curves whose dual graph is a tree with two vertices, one of which has genus $0$ and exactly two legs marked by $i$ and $n+1$.
\end{proof}

\medskip

In~\cite{BRS20}, the authors presented an explicit conjectural formula for a bihamiltonian structure of the DR hierarchy corresponding to a homogeneous CohFT. This in particular gives a recursion of certain type, called a {\it bihamiltonian recursion}, expressing the flows $\frac{\d}{\d t^\alpha_{d+1}}$, $1\le\alpha\le N$, of the hierarchy in terms of the flows $\frac{\d}{\d t^\alpha_d}$, $1\le\alpha\le N$. For a general homogeneous F-CohFT, we don't expect the corresponding DR hierarchy to have a Hamiltonian structure. However, we will now present a conjectural generalization of the bihamiltonian recursion in this setting.\\

Following~\cite{BRS20}, we associate with a differential polynomial $f\in\hcA$ a sequence of differential operators indexed by $\alpha=1,\ldots,N$ and $k\ge 0$: 
\begin{gather*}
L_\alpha^k(f):=\sum_{i\ge k}{i\choose k}\frac{\d f}{\d u^\alpha_i}\d_x^{i-k}.
\end{gather*}

\bigskip

Consider an arbitrary homogeneous F-CohFT and the corresponding DR hierarchy. Define an operator $R=(R^\alpha_\beta)$ by
\begin{gather*}
R^\alpha_\beta:=\hE_\gamma\left(L^0_\beta(g^{\alpha,0})\right)\circ\d_x+\left(\frac{1-\gamma}{2}\delta^\mu_\beta+q^\mu_\beta\right)L^0_\mu(g^{\alpha,0})_x+\d_x\circ L^1_\beta(g^{\alpha,0})\circ\d_x,
\end{gather*}
where the notation $\hE_\gamma\left(L^0_\beta(g^{\alpha,0})\right)$ (respectively, $L^0_\beta(g^{\alpha,0})_x$) means that we apply the operator~$\hE_\gamma$ (respectively, $\d_x$) to the coefficients of the operator $L^0_\beta(g^{\alpha,0})$.

\medskip

\begin{conjecture}\label{conjecture:recursion}
The following recursion relation is satisfied:
\begin{gather}\label{eq:DR-F-CohFT recursion}
R^\alpha_\mu P^\mu_{\beta,d}=\left(\left(d+\frac{3-\gamma}{2}\right)\delta^\mu_\beta+q^\mu_\beta\right)\d_x P^\alpha_{\mu,d+1}+(\d_xP^\alpha_{\mu,d})c^\mu_{\beta\nu}r^\nu,\qquad 1\le\alpha,\beta\le N,\quad d\ge -1.
\end{gather}
\end{conjecture}

\smallskip

\begin{proposition}
{\ }
\begin{enumerate}
\item If our homogeneous F-CohFT comes from a homogeneous CohFT, then the recursion~\eqref{eq:DR-F-CohFT recursion} coincides with the bihamiltonian recursion from part (2) of \cite[Conjecture~1.13]{BRS20}.
\item Conjecture~\ref{conjecture:recursion} is true in genus $0$, i.e., if we set $\eps=0$.
\end{enumerate}
\end{proposition}
\begin{proof}
For part (1), using the notations from paper~\cite{BRS20} let us note that $P^\alpha_{\beta,d}=\eta^{\alpha\mu}\frac{\delta\og_{\beta,d}}{\delta u^\mu}$. Therefore, we have to check that $\eta_{\beta\mu}K_2^{\alpha\mu}=R^\alpha_\beta$. This follows from the properties $\eta_{\beta\mu}\Omega^k(\og)^{\alpha\mu}=L^k_\beta(g^{\alpha,0})$ and $q^\mu_\alpha\eta_{\mu\beta}+\eta_{\alpha\mu}q^\mu_\beta=\gamma\eta_{\alpha\beta}$.\\ 

The proof of part~(2) follows closely the proof of \cite[Proposition~2.1]{BRS20}. 
\end{proof}

\medskip

%%%%%%%%%%%%%%%%%%%%%%%%%%%%%%%%%%%%%%%%
%%%%%%%%%%%%%%%%%%%%%%%%%%%%%%%%%%%%%%%%

\section{Principal hierarchy of a flat F-manifold and dispersive deformations}\label{section:principal hierarchy and dispersive deformations}

In this section, using the results from the previous section, we construct a family of dispersive integrable deformations of a principal hierarchy associated to an arbitrary semisimple flat F-manifold. Moreover, we prove that different hierarchies from this family are not equivalent to each other by a Miura transformation that is close to identity.\\

\subsection{Flat F-manifolds}\label{subsection:flat F-manifolds}

Here we recall the notion of a flat F-manifold (\cite{Get04,Man05}, see also \cite{AL18} and \cite{LPR09}) and its main properties.

\medskip

\begin{definition}
A {\it flat F-manifold} $(M,\nabla,\circ,e)$ is the datum of an analytic manifold $M$, an analytic connection $\nabla$ in the tangent bundle $T M$, an algebra structure $(T_p M,\circ)$ with unit~$e$ on each tangent space, analytically depending on the point $p\in M$, such that the one-parameter family of connections $\nabla_z=\nabla+z\circ$ is flat and torsionless for any $z\in\mbC$, and $\nabla e=0$.
\end{definition}

\medskip

The algebras $(T_pM,\circ)$ are commutative and associative. Let $t^\alpha$, $1\le\alpha\le N$, $N=\dim M$, be flat coordinates for the connection $\nabla$. Locally, there exist analytic functions $F^\alpha(t^1,\ldots,t^N)$, $1\leq\alpha\leq N$, such that the second derivatives 
\begin{gather}\label{eq:structure constants of flat F-man}
C^\alpha_{\beta\gamma}:=\frac{\d^2 F^\alpha}{\d t^\beta \d t^\gamma}
\end{gather}
are the structure constants of the algebras $(T_p M,\circ)$, $\frac{\d}{\d t^\beta}\circ\frac{\d}{\d t^\gamma}=C^\alpha_{\beta\gamma}\frac{\d}{\d t^\alpha}$. Also, in the coordinates~$t^\alpha$ the unit $e$ has the form $e=A^\alpha\frac{\d}{\d t^\alpha}$ for some constants $A^\alpha\in\mbC$. Moreover, the following equations are satisfied:
\begin{align}
&A^\mu\frac{\d^2 F^\alpha}{\d t^\mu\d t^\beta} = \delta^\alpha_\beta, && 1\leq \alpha,\beta\leq N,\label{eq:axiom1 of flat F-man}\\
&\frac{\d^2 F^\alpha}{\d t^\beta \d t^\mu} \frac{\d^2 F^\mu}{\d t^\gamma \d t^\delta} = \frac{\d^2 F^\alpha}{\d t^\gamma \d t^\mu} \frac{\d^2 F^\mu}{\d t^\beta \d t^\delta}, && 1\leq \alpha,\beta,\gamma,\delta\leq N,\label{eq:axiom2 of flat F-man}
\end{align}
which are often called the {\it oriented WDVV equations}. The $N$-tuple of functions $\oF=(F^1,\ldots,F^N)$ is called a {\it vector potential} of the flat F-manifold.\\

Conversely, given an open subset~$M$ of $\mbC^N$ and analytic functions $F^1,\ldots,F^N$ on~$M$ satisfying equations~\eqref{eq:axiom1 of flat F-man} and~\eqref{eq:axiom2 of flat F-man}, these functions define a flat F-manifold $(M,\nabla,\circ,A^\alpha\frac{\d}{\d t^\alpha})$ with the connection~$\nabla$ given by $\nabla_{\frac{\d}{\d t^\alpha}}\frac{\d}{\d t^\beta}=0$, and the multiplication $\circ$ given by the structure constants~\eqref{eq:structure constants of flat F-man}.\\

A point $p\in M$ of an $N$-dimensional flat F-manifold $(M,\nabla,\circ,e)$ is called \textit{semisimple} if $T_pM$ has a basis $\pi_1,\dots,\pi_N$ satisfying $\pi_\alpha \circ \pi_\beta = \delta_{\alpha,\beta} \pi_\alpha$. Moreover, locally around such a point one can choose coordinates $u^i$ such that $\frac{\d}{\d u^\alpha}\circ\frac{\d}{\d u^\beta}=\delta_{\alpha,\beta}\frac{\d}{\d u^\alpha}$. These coordinates are called {\it canonical coordinates}. In particular, this means that the set of semisimple points is open in~$M$. In the canonical coordinates, we have $e=\sum_\alpha\frac{\d}{\d u^\alpha}$. A flat F-manifold~$(M,\nabla,\circ,e)$ is called {\it semisimple} if the set of semisimple points is dense in~$M$.\\ 

A flat F-manifold given by a vector potential $\oF$ is called {\it homogeneous} if there exists a vector field $E$ of the form
\begin{equation}\label{eq:Euler of F-manifold}
E=(\underbrace{(\delta^\alpha_\beta-q^\alpha_\beta)t^\beta+r^\alpha}_{=:E^\alpha})\frac{\d}{\d t^\alpha},\qquad q^\alpha_\beta,r^\alpha\in\mbC,
\end{equation}
satisfying $[e,E]=e$ and such that
\begin{gather*}
E^\mu\frac{\d F^\alpha}{\d t^\mu}=(2\delta^\alpha_\beta-q^\alpha_\beta)F^\beta+A^\alpha_\beta t^\beta+B^\alpha
\end{gather*}
for some $A^\alpha_\beta,B^\alpha\in\mbC$. Note that this equation can be written more invariantly as $\Lie_E(\circ)=\circ$, where $\Lie_E$ denotes the Lie derivative. The vector field $E$ is called the {\it Euler vector field}. Around a semisimple point, the Euler vector field has the following form in canonical coordinates: $E=\sum_{i=1}^N(u^i+a^i)\frac{\d}{\d u^i}$ for some $a^i\in\mbC$.

\medskip

\begin{remark}
As we already mentioned in Remark~\ref{remark:homogeneous F-CohFTs}, our definition of a homogeneous flat F-manifold is slightly more general than the one from~\cite{ABLR20}, but all the results from that paper remains valid. 
\end{remark}

\smallskip

\begin{remark}
In~\cite{AL13}, the authors introduced the closely related notion of a {\it bi-flat F-manifold} that is the datum of two different flat F-manifold structures $(\nabla,\circ,e)$ and $(\nabla^{*},*,E)$ on the same manifold~$M$ intertwined by the following conditions: (1) $[e,E]=e$, $\Lie(\circ)=\circ$; (2) $X*Y:=(E\circ)^{-1}\,X\circ Y$ (or $X\circ Y=(e*)^{-1}X*Y$) for all local vector fields~$X, Y$ on~$M$, where $(E\circ)^{-1}$ is the inverse of the endomorphism of the tangent bundle given by $E\circ$; (3) $(d_{\nabla}-d_{\nabla^{*}})(X\,\circ)=0$ for all local vector fields~$X$ on~$M$, where $d_{\nabla}$ is the exterior covariant derivative. For a bi-flat F-manifold, the flat structure given by $(\nabla^*,*,E)$ is called the \emph{dual structure}. In the semisimple case, the flatness of the dual structure is equivalent to the condition $\nabla \nabla E=0$ \cite{AL17} (see~\cite{KMS18} for the regular case). Thus, in the structure of a semisimple homogeneous flat F-manifold is equivalent to the structure of a semisimple bi-flat F-manifold.
\end{remark}

\medskip

Given an F-CohFT $c_{g,n+1}\colon V^*\otimes V^{\otimes n} \to H^\even(\oM_{g,n+1})$, $\dim V=N$, and a basis $e_1,\ldots,e_N\in V$, with $e=A^\alpha e_\alpha$, an $N$-tuple of functions $(F^1,\ldots,F^N)$ satisfying equations~\eqref{eq:axiom1 of flat F-man} and~\eqref{eq:axiom2 of flat F-man} can be constructed as the following generating functions:
\begin{equation*}
F^\alpha(t^1,\ldots,t^N):=\sum_{n\geq 2}\frac{1}{n!}\sum_{1\leq\alpha_1,\ldots,\alpha_n\leq N}\left(\int_{\oM_{0,n+1}}c_{0,n+1}(e^\alpha\otimes\otimes_{i=1}^n e_{\alpha_i})\right)\prod_{i=1}^n t^{\alpha_i},
\end{equation*}
thus yielding an associated flat F-manifold structure on a formal neighbourhood of~$0$ in $V$ (see,  e.g.,~\cite[Proposition~3.2]{ABLR20}). The flat F-manifold associated to a homogeneous F-CohFT is homogeneous with the Euler vector field~\eqref{eq:Euler of F-manifold} where $q^\alpha_\beta e_\alpha:=Q e_\beta$ and $r^\alpha e_\alpha:=\br$.\\

\subsection{Principal hierarchy of a flat F-manifold}\label{subsection:principal hierarchy}

Given a flat F-manifold $(M,\nabla,\circ,e)$, one can construct an integrable dispersionless hierarchy called a {\it principal hierarchy} associated to $(M,\nabla,\circ,e)$ (see~\cite{LPR09}). This construction generalizes the notion of a principal hierarchy associated to a Dubrovin--Frobenius manifold. Before presenting the construction, let us introduce a small generalization of the space of densities of local multivector fields.\\

Let $U$ be an open subset of $\mbC^N$, with coordinates $u^1,\ldots,u^N$. Denote by $\mcO(U)$ the space of analytic functions on~$U$. Consider the following space:
$$
\hcA^{\bullet}_U:=\mcO(U)[u^*_{>0},\theta_{*,*}][[\eps]].
$$
Clearly, the space $\hcA^{\bullet}$ can be considered as the space $\hcA^{\bullet}_U$ where $U$ is a formal neighborhood of~$0$. The space~$\hcA^{\bullet}_U$ will also be called the space of densities of local multivector fields. It is easy to see that all the constructions from Section~\ref{subsection:formal loop space} (except, probably, the constructions related to the change of variables~\eqref{eq:u-p change}) work with the more general space $\hcA^{\bullet}_U$. The space of local multivector fields corresponding to $\hcA^\bullet_U$ will be denoted by $\hLambda^\bullet_U$.\\ 

Consider a flat F-manifold $(M,\nabla,\circ,e)$. For any point of~$M$, on its open neighbourhood~$U$, one can consider a basis (over $\mbC[[z]]$) $X_{\alpha}(z)=\sum_{d=-1}^{\infty}X_{\alpha,d}z^{d+1}$, $1\le\alpha\le N=\dim M$, in the space of flat sections of the deformed connection $\nabla_{-z}=\nabla-z\circ$: 
\begin{equation}\label{eq:defor flat}
0=(\nabla-z\circ)X_\alpha(z)=(\nabla-z\circ)\sum_{d=-1}^{\infty}X_{\alpha,d}z^{d+1}.
\end{equation}
It is immediate to see from \eqref{eq:defor flat} that $X_{\alpha, -1}$, $\alpha=1,\dots N$, are flat vector fields for $\nabla$, while the vector fields $X_{\alpha, d}$ are obtained via the recurrence relation $\nabla X_{\alpha, d+1}=X_{\alpha, d}\circ$. If $U$ is connected, then the collection of flat sections $X_\alpha(z)$ is determined uniquely up to a transformation of the form $X_\alpha(z)\mapsto X_\beta(z) G^\beta_\alpha(z)$, where $G(z)=(G^\alpha_\beta(z))\in\Mat_{N,N}(\mbC)[[z]]$ is invertible. If $M$ is simply connected, then flat sections $X_\alpha(z)$ can be constructed on the whole~$M$.

\medskip

\begin{definition}
A {\it calibration} of a flat F-manifold~$(M,\nabla,\circ,e)$ is a basis $X_{\alpha}=\sum_{d=-1}^{\infty}X_{\alpha,d}z^{d+1}$, $X_{\alpha,d}\in\cT(M)$, $1\le\alpha\le\dim M$, in the space of flat sections of the deformed connection $\nabla-z\circ$. A flat F-manifold with a fixed calibration is called a {\it calibrated flat F-manifold}.  
\end{definition}

\medskip

Consider now a flat F-manifold structure on $M\subset\mbC^N$ given by a vector potential $\oF$, together with a calibration $X_\alpha(z)$. The {\it principal hierarchy} associated to our calibrated flat F-manifold is the following system of PDEs:
\begin{equation}\label{eq: principal eq2}
\frac{\d u^{\alpha}}{\d t^{\beta}_{d}}=\d_x\left(\left.X^{\alpha}_{\beta,d}\right|_{t^\gamma=u^\gamma}\right),\qquad 1\le\alpha,\beta\le N,\quad d\ge 0,
\end{equation}
where $X^\alpha_{\beta,d}\frac{\d}{\d t^\alpha}:=X_{\beta,d}$. We see that the system~\eqref{eq: principal eq2} has the form of a system of conservation laws. Moreover, this is a system of quasilinear evolutionary PDEs, which is dispersionless and integrable, in the sense that all the flows pairwise commute (see \cite{LPR09}).\\ 

Suppose that $M$ is a formal neighbourhood of $0\in\mbC^N$. There exist unique flat sections $X_\alpha(z)$ on~$M$ satisfying the condition $X_{\alpha,-1}=\frac{\d}{\d t^\alpha}$ and the condition that~$X_{\alpha,d}$ vanish at $0$ for $d\ge 0$. The corresponding principal hierarchy is called the {\it ancestor principal hierarchy}.

%Around a semisimple point, in canonical coordinates $\rho^i$, the principal hierarchy \eqref{eq: principal eq2} has the diagonal form:
%\begin{equation*}
%\frac{\partial \rho^i}{\partial t^{\beta}_d}=X^i_{\beta,d-1}\rho^i_x,\qquad 1\le i,\beta\le N,\quad d\ge 0,
%\end{equation*}
%where $\sum_{i=1}^N X^i_{\beta,d}\frac{\d}{\d\rho^i}:=X_{\beta,d}$. So the canonical coordinates are the Riemann invariants of the quasilinear system \eqref{eq: principal eq2}.\\

\medskip

\begin{proposition}
Consider an F-CohFT and the associated flat F-manifold and the DR hierarchy. Then the dispersionless part of the DR hierarchy coincides with the ancestor principal hierarchy of the flat F-manifold.
\end{proposition}
\begin{proof}
This immediately follows from the construction of the DR hierarchy and \cite[Proposition~3.2]{ABLR20} (see also an analogous statement in~\cite[Section~4.2.2]{Bur15}).
\end{proof}

\medskip

We see that this proposition can be immediately used for a construction of dispersive deformations of ancestor principal hierarchies. In order to construct dispersive deformations of arbitrary principal hierarchies, we need a generalization of the construction of the DR hierarchy, which we will introduce in the next section.\\  

\subsection{Dispersive deformations of a principal hierarchy: descendant DR hierarchies}\label{subsection:deformations of principal hierarchy}

In order to construct dispersive deformations of a principal hierarchy associated to an arbitrary semisimple flat F-manifold, we first need to study analytic families of F-CohFTs depending on a semisimple point of a flat F-manifold.\\

Consider a semisimple flat F-manifold structure on $M\subset\mbC^N$ defined by a vector potential~$\oF$. Recall that for an arbitrary semisimple point on $M$, on its connected open neighborhood~$U$, one has the following objects (we use the notations from~\cite[Section~1.2]{ABLR20}): 
\begin{itemize}
\item canonical coordinates $u^i$;
\item the matrix $\tPsi:=\big(\frac{\d u^i}{\d t^\alpha}\big)$;
\item the matrices $\tD$ and $\tGamma$ defined by $d\tPsi\cdot\tPsi^{-1}=\tD+[\tGamma,dU]$, where $\tD$ is a diagonal matrix consisting of one-forms, $\tGamma$ is a matrix with vanishing diagonal entries, and $U:=\diag(u^1,\ldots,u^N)$ (in the homogeneous case this is the operator of multiplication by the Euler vector field);
\item a diagonal nondegenerate matrix $H=\diag(H_1,\ldots,H_N)$ defined by $dH\cdot H^{-1}=-\tD$ (the entries of this matrix can be interpreted as the Lam\'e coefficients  of a diagonal metric associated with the flat F-manifold); 
\item the matrices $\Psi$ and $\Gamma$ defined by $\Psi:=H\tPsi$ and $\Gamma:=H\tGamma H^{-1}$;
\item a sequence of matrices $R_0=\Id,R_1,R_2,\ldots$ defined by the relations $dR_{k-1}+R_{k-1}[\Gamma,dU]=[R_k,dU]$, $k\ge 1$.
\end{itemize} 
Note that the matrix $H$ is defined uniquely up to the transformation $H\mapsto A H$, where $A$ is a constant nondegenerate diagonal matrix. After such a transformation, the matrices $\Psi$, $\Gamma$, and~$R_k$ transform as follows: $\Psi\mapsto A \Psi$, $\Gamma\mapsto A \Gamma A^{-1}$, $R_k\mapsto A R_k A^{-1}$. Recall also that if we fix~$H$, then the matrices $R_k$ are defined uniquely up to the transformation
\begin{gather}\label{eq:transformation of the R-matrix}
\Id+\sum_{i\ge 1}R_iz^i\mapsto\bigg(\Id+\sum_{i\ge 1}D_iz^i\bigg)\bigg(\Id+\sum_{i\ge 1}R_iz^i\bigg),
\end{gather}
where $D_i$, $i\ge 1$, are arbitrary constant diagonal matrices.\\

Using the notations from~\cite[Section~4.4]{ABLR20}, for any $G_0\in\mbC^N$, let us define an analytic family of F-CohFTs parameterized by a point $\ot\in U$ by
$$
c^{G_0,\ot}:=\tPsi^{-1}H^{-1}R^{-1}(-z)H.c^{\triv,H^{-2}G_0}, 
$$ 
where $R(z):=\sum_{i\ge 0}R_iz^i$. Note that if $G_0=0$, then the maps $c^{G_0,\ot}_{g,n+1}$ are zero for $g\ge 1$.\\

Let $\tau^1,\ldots,\tau^N$ be formal variables. Recall from~\cite[Section~3.2]{ABLR20} (note, however, that we prefer to use a different notation here) that for an F-CohFT $c_{g,n+1}\colon V^*\otimes V^{\otimes n}\to H^\even(\oM_{g,n+1})$ its formal shift $\Sh_{\otau}(c)_{g,n+1}\colon V^*\otimes V^{\otimes n}\to H^\even(\oM_{g,n+1})[[\tau^*]]$ is defined by
\begin{gather*}
\Sh_{\otau}(c)_{g,n+1}:=\sum_{m\geq 0} \frac{1}{m!}\pi_{m*}\circ c_{g,n+m+1}\circ \left(\otimes (\tau^\alpha e_\alpha)^{\otimes m}\right),
\end{gather*}
where $\otau=(\tau^1,\ldots,\tau^N)$ and $\pi_m\colon\oM_{g,n+m+1}\to\oM_{g,n+1}$ forgets the last $m$ marked points. The maps $\Sh_{\otau}(c)_{g,n+1}$ form an F-CohFT with the coefficients in~$\mbC[[\tau^*]]$.

\medskip

\begin{proposition}\label{proposition:family of F-CohFTs}
1. For any $\ot_0=(t_0^1,\ldots,t_0^N)\in U$, a vector potential of the flat F-manifold corresponding to the F-CohFT $c^{G_0,\ot_0}$ is equal to $\oF(t^*-t^*_0)$.\\
2. For any fixed $\ot_0\in U$, the Taylor expansion of $c^{G_0,\ot}$ at $\ot_0$ coincides with the formal shift of~$c^{G_0,\ot_0}$, i.e., $c^{G_0,\ot_0+\otau}_{g,n+1}=\Sh_{\otau}(c^{G_0,\ot_0})_{g,n+1}$, as elements of $\Hom\left(V^*\otimes V^{\otimes n},H^\even(\oM_{g,n+1})[[\tau^*]]\right)$.
\end{proposition}
\begin{proof}
{\it 1}. We know that under the transformation $H\mapsto AH$, where $A$ is a nondegenerate constant diagonal matrix, $R(z)$ transforms as $R(z)\mapsto AR(z)A^{-1}$, and therefore
$$
\tPsi^{-1}H^{-1}R^{-1}(-z)H.c^{\triv,H^{-2}G_0}\mapsto \tPsi^{-1}H^{-1}R^{-1}(-z)H.c^{\triv,H^{-2}(A^{-2}G_0)}. 
$$
Thus, for a fixed $\ot$ the family $\{c^{G_0,\ot}\}_{G_0\in\mbC^N}$ doesn't depend on a choice of~$H$. Let us choose~$H$ such that $H_i(\ot_0)=1$, then $c^{G_0,\ot_0}=\Psi^{-1}(\ot_0)R^{-1}(-z,\ot_0).c^{\triv,G_0}$. The fact that a vector potential of the associated flat F-manifold is equal to $\oF(t^*-t_0^*)$ was proved in~\cite[Section~4.4]{ABLR20} (see equation~(4.3) there).\\

{\it 2}. An elementary computation shows that $c^{G_0,\ot}=\Psi^{-1}R^{-1}(-z).c^{\oH,H^{-1}G_0}$, where by $\oH$ we denote the vector $(H_1,\ldots,H_N)$. The statement of part 2 of the proposition is equivalent to the property
\begin{gather*}
\frac{\d}{\d t^\beta}\left(\Psi^{-1}R^{-1}(-z).c^{\oH,H^{-1}G_0}\right)_{g,n+1}=\pi_{1*}\circ\left(\Psi^{-1}R^{-1}(-z).c^{\oH,H^{-1}G_0}\right)_{g,n+2}\circ(\otimes e_\beta),
\end{gather*}
which was proved in~\cite[proof of Proposition~4.11]{ABLR20}.
\end{proof}

\medskip

To our family of F-CohFTs $c^{G_0,\ot}$, $\ot\in U$, one can associate a natural vector field $\cX=\cX^\alpha\frac{\d}{\d t^\alpha}$ on~$U$ where $\cX^\alpha$ is the degree zero part of $c^{G_0,\ot}_{1,1}(e^\alpha)\in H^*(\oM_{1,1})$. Note that $\cX^\alpha=\sum_{i=1}^N\frac{\d t^\alpha}{\d u^i}H_i^{-2}G_0^i$. This motivates the following definition.

\medskip

\begin{definition}
Consider a semisimple flat F-manifold $(M,\nabla,\circ,e)$. A vector field $\cX$ on $M$ is called a {\it framing} if around each semisimple point of~$M$, in canonical coordinates $u^i$, the field $\cX$ has the form $\cX=\sum_{i=1}^N\alpha_i H_i^{-2}\frac{\d}{\d u^i}$ for some complex constants $\alpha_i$, $1\le i\le N$.
\end{definition}

\medskip

Using this language, we can say that our family of F-CohFTs $c^{G_0,\ot}$ induces a framing on $U$.\\

Suppose that all the points of our flat F-manifold $M$ are semisimple and $\cX$ is a framing on~$M$. We see that for any point $\ot_0\in M$ the above construction gives a family of F-CohFTs around~$\ot_0$ such that the induced framing coincides with~$\cX$. This family is not unique because the matrix~$R(z)$ is defined uniquely only up to the transformation~\eqref{eq:transformation of the R-matrix}. Suppose that $M$ is simply connected. Then it is easy to see that there is a consistent choice of matrix $R(z)$ in all the charts such that the local families glue in a global family of F-CohFTs parameterized by $\ot\in M$. Let us denote this global family by $c^{\cX,\ot}$. This global family is not unique: in order to fix the ambiguity, one can, for example, fix a choice of matrix $R(z)$ at some fixed point of $M$. Note that if $\cX=0$, then the maps $c^{\cX,\ot}_{g,n+1}$ are zero for $g\ge 1$.\\

Let us now apply the construction of the DR hierarchy to the F-CohFTs~$c^{\cX,\ot}$. We obtain a family of densities $Y_{\beta,d}^\ot\in(\hcA^1)^{[1]}$, where the superscript $\ot$ signals that the densities $Y^\ot_{\beta,d}$ analytically depend on $\ot\in M$. It is convenient to consider the generating series of densities~$Y^\ot_{\beta,d}$: 
$$
Y^\ot_\beta(z):=\sum_{d\ge -1}Y^\ot_{\beta,d}z^{d+1}.
$$

\smallskip

\begin{lemma}\label{lemma:property of family of densities}
We have $\frac{\d Y^\ot_\beta(z)}{\d t^\gamma}=\frac{\d Y^\ot_\beta(z)}{\d u^\gamma}-C^\mu_{\beta\gamma}z Y^\ot_\mu(z)$, $1\le\beta,\gamma\le N$.
\end{lemma}
\begin{proof}
This follows from the definition of the densities $Y^\ot_{\beta,d}$, the property $\frac{\d}{\d t^\beta}(c^{\cX,\ot})_{g,n+1}=\pi_{1*}\circ(c^{\cX,\ot})_{g,n+2}\circ(\otimes e_\beta)$ (which is equivalent to part~2 of Proposition~\ref{proposition:family of F-CohFTs}), and the fact that $\pi_1^*\psi_i^d = \psi_i^d-\delta^0_{i,n+1}\pi_1^*\psi_i^{d-1}$, $1\le i\le n$, $d\ge 1$, where the class~$\delta^0_{i,n+1}$ was defined in the proof of Proposition~\ref{proposition:homogeneous DR hierarchy}.
\end{proof}

\medskip

Consider now a calibration $X_\alpha(z)$ of our flat F-manifold~$M$. Define densities $\tY^\ot_{\beta,d}\in(\hcA^1)^{[1]}$, $1\le\beta\le N$, $d\ge -1$, by $\sum_{d\ge -1}\tY^\ot_{\beta,d}z^{d+1}:=\tY^\ot_\beta(z)$ where
$$
\tY^\ot_\beta(z):=Y_\mu^\ot(z)X^\mu_\beta(z).
$$

\smallskip

\begin{lemma}
We have $\frac{\d\tY^\ot_{\beta,d}}{\d t^\gamma}=\frac{\d\tY^\ot_{\beta,d}}{\d u^\gamma}$.
\end{lemma}
\begin{proof}
This immediately follows from Lemma~\ref{lemma:property of family of densities} and the property $\frac{\d X^\mu_\beta(z)}{\d t^\gamma}=C^\mu_{\gamma\nu}z X^\nu_\beta(z)$.
\end{proof}

\medskip

Define densities of vector fields $Y^\desc_{\beta,d}\in(\hcA^1_M)^{[1]}$, $1\le\beta\le N$, $d\ge -1$, by
$$
Y^\desc_{\beta,d}:=\left.\left(\left.\tY^\ot_{\beta,d}\right|_{u^*=0}\right)\right|_{t^\gamma=u^\gamma}\in(\hcA^1_M)^{[1]}.
$$
The previous lemma implies that for a fixed $\ot\in M$ the density $\tY^\ot_{\beta,d}$ is the Taylor expansion of the density $Y^\desc_{\beta,d}$ at $u^\gamma=t^\gamma$, i.e., $\tY^\ot_{\beta,d}=\left.Y^\desc_{\beta,d}\right|_{u^\gamma\mapsto t^\gamma+u^\gamma}$, as elements of $(\hcA^1)^{[1]}$. Therefore, since for any $\ot\in M$ the densities $\tY^\ot_{\beta,d}$ produce a hierarchy of pairwise commuting flows, the densities $Y^\desc_{\beta,d}$ also produce a hierarchy of pairwise commuting flows. This hierarchy is called the {\it descendant DR hierarchy}.\\

In more details, the equations of the descendant DR hierarchy are given by
$$
\frac{\d u^\alpha}{\d t^\beta_d}=\d_x P^{\desc;\alpha}_{\beta,d},\qquad 1\le\alpha,\beta\le N,\quad d\ge 0,
$$
where $P^{\desc;\alpha}_{\beta,d}=\left.\left(\left.\tP^{\ot;\alpha}_{\beta,d}\right|_{u^*=0}\right)\right|_{t^\gamma=u^\gamma}$, $\tP^{\ot;\alpha}_{\beta,d}=\sum_{i=0}^{d+1}P^{\ot;\alpha}_{\mu,d-i}X^\mu_{i-1,\beta}$, and $P^{\ot;\alpha}_{\beta,d}$ are the differential polynomials~\eqref{eq:P-polynomials for DR hierarchy} corresponding to the F-CohFT~$c^{\cX,\ot}$. Also, we adopt the convention $P^{\desc;\alpha}_{\beta,-1}:=X^\alpha_{\beta,-1}$. Note that we have $\frac{\d\tP^{\ot;\alpha}_{\beta,d}}{\d t^\gamma}=\frac{\d\tP^{\ot;\alpha}_{\beta,d}}{\d u^\gamma}$.\\ 

We immediately see that $\left.P^{\desc;\alpha}_{\beta,d}\right|_{\eps=0}=\left.X^\alpha_{\beta,d}\right|_{t^\gamma=u^\gamma}$, and therefore the dispersionless part of the descendant DR hierarchy coincides with the principal hierarchy. For $\cX=0$, the descendant DR hierarchy coincides with the principal hierarchy. \\

Statements analogous to the ones from Theorem~\ref{theorem:Main} are true for the descendant DR hierarchy. We present here the proof of a couple of them.\\

Note that if $X^\alpha_{\beta,-1}=\delta^\alpha_\beta$, then $X^\alpha_{\un,0}$ coincides with $t^\alpha$ up to a constant. We will say that a calibration is of {\it standard type} if $X^\alpha_{\beta,-1}=\delta^\alpha_\beta$ and $X^\alpha_{\un,0}=t^\alpha$.

\medskip

\begin{proposition}
1. We have $\frac{\d}{\d u^\un}P^{\desc;\alpha}_{\beta,d+1}=P^{\desc;\alpha}_{\beta,d}$, $1\le\alpha,\beta\le N$, $d\ge -1$.\\
2. If our calibration is of standard type, then $\frac{\d}{\d u^\beta}P^{\desc;\alpha}_{\un,1}=D P^{\desc;\alpha}_{\beta,0}$, $1\le\alpha,\beta\le N$.
\end{proposition}
\begin{proof}
To prove part 1, we compute $\left.\frac{\d}{\d u^\un}P^{\desc;\alpha}_{\beta,d+1}\right|_{u^\gamma\mapsto t^\gamma+u^\gamma}=\frac{\d}{\d u^\un}\tP^{\ot;\alpha}_{\beta,d+1}=\frac{\d}{\d u^\un}\sum_{i=0}^{d+2}P^{\ot;\alpha}_{\mu,d+1-i}X^\mu_{i-1,\beta}=\sum_{i=0}^{d+1}P^{\ot;\alpha}_{\mu,d-i}X^\mu_{i-1,\beta}=\tP^{\ot;\alpha}_{\beta,d}=\left.P^{\desc;\alpha}_{\beta,d}\right|_{u^\gamma\mapsto t^\gamma+u^\gamma}$.\\

For part 2, we compute $\left.\frac{\d}{\d u^\beta}P^{\desc;\alpha}_{\un,1}\right|_{u^\gamma\mapsto t^\gamma+u^\gamma}=\frac{\d}{\d u^\beta}\tP^{\ot;\alpha}_{\un,1}=\frac{\d}{\d u^\beta}\left(P^{\ot;\alpha}_{\un,1}+P^{\ot;\alpha}_{\mu,0}t^\mu\right)=D P^{\ot;\alpha}_{\beta,0}+t^\mu\frac{\d}{\d u^\mu}P^{\ot;\alpha}_{\beta,0}=\left.D P^{\desc;\alpha}_{\beta,0}\right|_{u^\gamma\mapsto t^\gamma+u^\gamma}$.
\end{proof}

\medskip

To summarize the above constructions, given the following data:
\begin{itemize}
\item a flat F-manifold structure on $M\subset\mbC^N$ given by a vector potential $\oF$ such that $M$ is simply connected and all the points of $M$ are semisimple;
\item its calibration;
\item a framing on $M$;
\end{itemize}
we have constructed a dispersive integrable deformation of the principal hierarchy . In the next section, we will prove that the dispersive deformations corresponding to different framings are not related to each other by a Miura transformation that is close to identity.\\ 

\subsection{Nonequivalence of dispersive deformations}

We say that a Miura transformation~\eqref{eq:Miura1}--\eqref{eq:Miura2} is {\it close to identity} if $\tu^\alpha|_{\eps=0}=u^\alpha$. 

\medskip

\begin{definition}\label{definition:equivalent deformations}
Two dispersive deformations of the principal hierarchy of a calibrated flat F-manifold are called {\it equivalent} if they are related by a Miura transformation that is close to identity.
\end{definition}

\smallskip

\begin{theorem}\label{theorem:nonequivalent deformations}
Let us fix a calibrated flat F-manifold structure on a simply connected open subset $M\subset\mbC^N$, with a vector potential~$\oF$ and which is semisimple at each point of~$M$. Then, for different framings~$\cX$ and~$\hcX$ on~$M$, the corresponding descendant DR hierarchies are not equivalent. 
\end{theorem}
\begin{proof}[Proof of Theorem~\ref{theorem:nonequivalent deformations}]
Following~\cite{AL18}, for a system of evolutionary PDEs of the form
\begin{gather*}
\frac{\d u^\alpha}{\d t}=Q^\alpha,\qquad Q^\alpha\in\hcA^{[1]}_M,\quad 1\le\alpha\le N,
\end{gather*}
let us consider the associated {\it Miura matrix} $S(z)=(S^\alpha_\beta(z))\in\Mat_{N,N}\left(\mcO(M)[[z]]\right)$ defined by
$$
S^\alpha_\beta(z):=\sum_{d\ge 0}\left.\frac{\d Q^\alpha}{\d u^\beta_{d+1}}\right|_{\substack{u^\gamma_c=\delta_{c,0}t^\gamma\\\eps=z}}.
$$
For a Miura transformation~\eqref{eq:Miura1}--\eqref{eq:Miura2} that is close to identity, introduce its {\it symbol} $T(z)=(T^\alpha_\beta(z))\in\Mat_{N,N}\left(\mcO(M)[[z]]\right)$ by
$$
T^\alpha_\beta(z):=\sum_{d\ge 0}\left.\frac{\d\tu^\alpha}{\d u^\beta_d}\right|_{\substack{u^\gamma_c=\delta_{c,0}t^\gamma\\\eps=z}}.
$$
It is easy to see that under the Miura transformation the Miura matrix of our system of PDEs transforms as follows: 
\begin{gather*}
S(z)\mapsto T(z) S(z) T(z)^{-1}.
\end{gather*}

\bigskip

Now consider the descendant DR hierarchies corresponding to different framings~$\cX$ and~$\hcX$. Let us denote the Miura matrices of a flow $\frac{\d}{\d t^\alpha_d}$ from these two hierarchies by~$S^{(\alpha,d)}(z)$ and~$\hS^{(\alpha,d)}(z)$, respectively. Clearly, $S^{(\alpha,d)}(0)=\hS^{(\alpha,d)}(0)$. Suppose that the hierarchies are related by a Miura transformation that is close to identity. Denote its symbol by $T(z)$. For the calibration of our flat F-manifold, without loss of generality, we can assume that $X^\alpha_{\beta,-1}=\delta^\alpha_\beta$. Consider the expansions $S^{(\alpha,d)}(z)=\sum_{i\ge 0}S^{(\alpha,d)}_{2i}z^{2i}$, $\hS^{(\alpha,d)}(z)=\sum_{i\ge 0}\hS^{(\alpha,d)}_{2i}z^{2i}$, $T(z)=\sum_{i\ge 0}T_i z^i$. Then we have
\begin{align}
\hS^{(\alpha,d)}(z)=T(z) S^{(\alpha,d)}(z)T(z)^{-1}\Rightarrow& 
\left\{\begin{aligned}
&[T_1,S^{(\alpha,d)}_0]=0,\\
&S^{(\alpha,d)}_2+[S^{(\alpha,d)}_0,T_1]T_1+[T_2,S^{(\alpha,d)}_0]=\hS^{(\alpha,d)}_2,
\end{aligned}\right.\notag\\
\Rightarrow&
S^{(\alpha,d)}_2-\hS^{(\alpha,d)}_2=[S^{(\alpha,d)}_0,T_2].\label{eq:S and T}
\end{align}

\bigskip

For the descendant DR hierarchy corresponding to the framing $\cX$, we have
$$
\tP^{\ot;\alpha}_{\un,1}=P^{\ot;\alpha}_{\un,1}+P^{\ot;\alpha}_{\mu,0}X^\mu_{\un,0}+X^\alpha_{\un,1},\qquad \tP^{\ot;\alpha}_{\mu,0}=P^{\ot;\alpha}_{\mu,0}+X^\alpha_{\mu,0},
$$
which implies that the matrix $S=(S^\alpha_\beta):=S^{(\un,1)}_2-\sum_{\mu=1}^NX^\mu_{\un,0}S^{(\mu,0)}_2$ is given by
\begin{multline*}
S^\alpha_\beta=\left.\Coef_{\eps^2}\frac{\d P^{\ot;\alpha}_{\un,0}}{\d u^\beta_2}\right|_{u^*=0}=
\Coef_{a^2}\int_{\DR_1(a,0,-a)}\lambda_1\psi_2 c^{\cX,\ot}_{1,3}(e^\alpha\otimes e\otimes e_\beta)=\\
=2\Coef_{a^2}\int_{\DR_1(a,-a)}\lambda_1 c_{1,2}^{\cX,\ot}(e^\alpha\otimes e_\beta),
\end{multline*}
which is equal to $2\Coef_{a^2}\int_{\DR_1(a,-a)}\lambda_1=\frac{1}{12}$ times the degree zero part of $c_{1,2}^{\cX,\ot}(e^\alpha\otimes e_\beta)$. By the construction of the cohomological field theory $c^{\cX,\ot}$, the degree zero part of $c_{1,2}^{\cX,\ot}(e^\alpha\otimes e_\beta)$ is equal to $\sum_{i=1}^N\frac{\d t^\alpha}{\d\hu^i}\cX^i\frac{\d\hu^i}{\d t^\beta}$, where $\sum_{i=1}^N\cX^i\frac{\d}{\d \hu^i}:=\cX$ and $\hu^i$ are local canonical coordinates on $M$. We conclude that if we denote $\hS:=\hS^{(\un,1)}_2-\sum_{\mu=1}^N X^\mu_{\un,0}\hS^{(\mu,0)}_2$, then $\tPsi(S-\hS)\tPsi^{-1}$ is a nonzero diagonal matrix. On the other hand, since $\tPsi S^{(\alpha,d)}_0\tPsi^{-1}$ is a diagonal matrix for any $1\le\alpha\le N$ and $d\ge 0$, the diagonal part of $\left[\tPsi\left(S^{(\un,1)}_0-\sum_{\mu=1}^N X^\mu_{\un,0}S_0^{(\mu,0)}\right)\tPsi^{-1},\tPsi T_2\tPsi^{-1}\right]$ is equal to zero, which contradicts~\eqref{eq:S and T}.
\end{proof}

\medskip

Note that during the proof of the theorem we have obtained the following explicit relation between a framing and the differential polynomials defining the flows $\frac{\d}{\d t^\un_1}$ and $\frac{\d}{\d t^\mu_0}$ of a corresponding descendant DR hierarchy.

\medskip

\begin{lemma}\label{lemma:framing and DR}	
Consider a flat F-manifold, a calibration satisfying $X^{\alpha}_{\beta,-1}=\delta^\alpha_\beta$, a framing~$\cX=\cX^\alpha\frac{\d}{\d t^\alpha}$, and a corresponding descendant DR hierarchy. Then we have 
$$
\cX^\alpha=\left.12\frac{\d}{\d u^\un_{xx}}\Coef_{\eps^2}\left(P^{\desc;\alpha}_{\un,1}-P^{\desc;\alpha}_{\beta,0}X^\beta_{\un,0}\right)\right|_{u^*=t^*}.
$$
\end{lemma}  

\medskip

\subsection{Homogeneous dispersive deformations}

As at the beginning of Section~\ref{subsection:deformations of principal hierarchy}, consider a semisimple flat F-manifold structure on $M\subset\mbC^N$ defined by a vector potential $\oF$, a semisimple point, canonical coordinates~$u^i$ on an open neighborhood~$U$ of this point, the diagonal matrix of one-forms $\tD$, a diagonal nondegenerate matrix $H$, and matrices $R_k$. Suppose that our flat F-manifold is homogeneous with an Euler vector field~$E$ of the form~\eqref{eq:Euler of F-manifold}. By~\cite[Proposition~1.14]{ABLR20}, the diagonal matrix~$i_E\tD$ is constant, $i_E\tD=-\diag(\delta_1,\ldots,\delta_N)=-\Delta$, $\delta_i\in\mbC$. Moreover, we have $E^\alpha\frac{\d}{\d t^\alpha} H=\Delta H$, and by~\cite[Proposition~1.16]{ABLR20} we can fix a choice of matrices $R_k$ by the additional conditions $E^\alpha\frac{\d}{\d t^\alpha} R_k=-k R_k+[\Delta,R_k]$ for $k\ge 1$. By~\cite[proof of Theorem~4.10]{ABLR20}, for an arbitrary $1\le l\le N$ and an eigenvector $G_0$ of the matrix~$\Delta$ corresponding to the eigenvalue $\delta_l$ the family of F-CohFTs $c^{G_0,\ot}$ satisfies the property
\begin{gather*}
\Deg\circ c^{G_0,\ot}_{g,n+1}+E^\alpha\frac{\d}{\d t^\alpha}c^{G_0,\ot}_{g,n+1}=c^{G_0,\ot}_{g,n+1}\circ\left(-Q^t\otimes\Id^{\otimes n}+\hspace{-0.2cm}\sum_{i+j=n-1}\hspace{-0.2cm}\Id\otimes\Id^{\otimes i}\otimes Q\otimes \Id^{\otimes j}\right)-2\delta_l g c_{g,n+1}^{G_0,\ot}.
\end{gather*}
This implies that for any $\ot\in U$ the F-CohFT $c^{G_0,\ot}$ is homogeneous of conformal dimension~$-2\delta_l$. Note that the corresponding framing $\cX$ on $U$ satisfies the property $[E,\cX]=(-2\delta_l-1)\cX$.\\  

Suppose that~$M$ is connected, then it is clear that up to permutations of the components the vector $(\delta_1,\ldots,\delta_N)$ doesn't depend on a semisimple point. We come to the following natural definition.

\medskip

\begin{definition}
The vector $\ogamma:=(-2\delta_1,\ldots,-2\delta_N)$ is called the {\it vector of conformal dimensions} corresponding to our flat F-manifold.
\end{definition}

\medskip

Suppose that all the points of $M$ are semisimple. As in the previous section, we can now glue the local families of F-CohFTs in a global family. Note that given a framing $\cX$ on $M$ satisfying $[E,\cX]=(-2\delta_l-1)\cX$ we can now construct a unique global family $c^{\cX,\ot}$, $\ot\in M$, of F-CohFTs fixing the choice of matrices $R_k$ using the Euler vector field.\\

Summarizing the considerations of this section, we obtain the following result.

\medskip

\begin{theorem}
Consider a homogeneous flat F-manifold structure on a connected open subset $M\subset\mbC^N$ defined by a vector potential $\oF$. Suppose that all the points of $M$ are semisimple. Let $\ogamma=(\gamma_1,\ldots,\gamma_N)$ be the vector of conformal dimensions. Let $1\le l\le N$ and let $\cX$ be a framing on $M$ such that $[E,\cX]=(\gamma_l-1)\cX$. Then the family of F-CohFTs $c^{\cX,\ot}$ satisfies the property
\begin{gather*}
\Deg\circ c^{\cX,\ot}_{g,n+1}+E^\alpha\frac{\d}{\d t^\alpha}c^{\cX,\ot}_{g,n+1}=c^{\cX,\ot}_{g,n+1}\circ\left(-Q^t\otimes\Id^{\otimes n}+\hspace{-0.2cm}\sum_{i+j=n-1}\hspace{-0.2cm}\Id\otimes\Id^{\otimes i}\otimes Q\otimes \Id^{\otimes j}\right)+\gamma_l g c_{g,n+1}^{\cX,\ot}.
\end{gather*}
In particular, for any $\ot\in M$ the F-CohFT $c^{\cX,\ot}$ is homogeneous of conformal dimension $\gamma_l$.
\end{theorem}

\medskip

Let us now discuss properties of the descendant DR hierarchies in the homogeneous case. Under the assumptions of the theorem, suppose also that $M$ is simply connected. By~\cite[Proposition~4.4]{BB19}, there exists a calibration $X_\alpha(z)$ and complex matrices $\tR_i$, $i\ge 1$, such that $X^\alpha_{\beta,-1}=\delta^\alpha_\beta$, $[Q,\tR_i]=i\tR_i$, and 
$$
E^\mu\frac{\d}{\d t^\mu}X(z)=z\frac{\d}{\d z}X(z)+[X(z),Q]+X(z)\tR(z),
$$ 
where $X(z):=\sum_{d\ge -1}\left(X^\alpha_{\beta,d}\right)z^{d+1}$ and $\tR(z):=\sum_{i\ge 1}\tR_i z^i$. Such a calibration is called {\it homogeneous}. Consider now the associated descendant DR hierarchy.\\

Let us introduce a generating series $P^\desc(z)$ by
$$
P^\desc(z):=\sum_{d\ge -1}\left(P^{\desc;\alpha}_{\beta,d}\right)z^{d+1}. 
$$

\medskip

\begin{proposition}
We have 
$$
\hE_{\gamma_l}P^\desc(z)=z\frac{\d}{\d z}P^\desc(z)+[P^\desc(z),Q]+P^\desc(z)\tR(z).
$$
\end{proposition}
\begin{proof}
Let us introduce generating series $\tP^\ot(z)$ and $P^\ot(z)$ by
$$
\tP^\ot(z):=\sum_{d\ge -1}\left(\tP^{\ot;\alpha}_{\beta,d}\right)z^{d+1},\qquad P^\ot(z):=\sum_{d\ge -1}\left(P^{\ot;\alpha}_{\beta,d}\right)z^{d+1}.
$$
We have to check that
$$
\left.\hE_{\gamma_l}P^\desc(z)\right|_{u^\gamma\mapsto t^\gamma+u^\gamma}=z\frac{\d}{\d z}\tP^{\ot}(z)+[\tP^{\ot}(z),Q]+\tP^{\ot}(z)\tR(z).
$$
For this, we compute
\begin{align}
\left.\hE_{\gamma_l}P^\desc(z)\right|_{u^\gamma\mapsto t^\gamma+u^\gamma}=&\left(\hE_{\gamma_l}+t^\alpha(\delta_\alpha^\beta-Q^\beta_\alpha)\frac{\d}{\d u^\beta}\right)\tP^\ot(z)=\notag\\
=&\left(\hE_{\gamma_l}+t^\alpha(\delta_\alpha^\beta-Q^\beta_\alpha)\frac{\d}{\d u^\beta}\right)P^\ot(z)\cdot X(z)\xlongequal{\text{by Prop.~\ref{proposition:homogeneous DR hierarchy}}}\notag\\
=&z\frac{\d}{\d z}P^\ot(z)\cdot X(z)+[P^\ot(z),Q]X(z)+z P^\ot(z)(E^\gamma C_\gamma)X(z),\label{eq:tmp for homogeneity of descendant DR}
\end{align}
where $C_\gamma:=(C^\alpha_{\gamma\beta})$, and we recall that $C^\alpha_{\beta\gamma}=\frac{\d^2 F^\alpha}{\d t^\beta\d t^\gamma}$. Since $z E^\gamma C_\gamma X(z)=E^\gamma\frac{\d}{\d t^\gamma}X(z)=z\frac{\d}{\d z}X(z)+[X(z),Q]+X(z)\tR(z)$, the expression in line~\eqref{eq:tmp for homogeneity of descendant DR} is equal to
\begin{align*}
&z\frac{\d}{\d z}P^\ot(z)\cdot X(z)+[P^\ot(z),Q]X(z)+P^\ot(z)\left(z\frac{\d}{\d z}X(z)+[X(z),Q]+X(z)\tR(z)\right)=\\
=&z\frac{\d}{\d z}\tP^{\ot}(z)+[\tP^{\ot}(z),Q]+\tP^{\ot}(z)\tR(z),
\end{align*}
as required.
\end{proof}

\bigskip

%%%%%%%%%%%%%%%%%%%%%%%%%%%%%%%%%%%%%%%%%%%%%%%%%%%%%%%%%%%%%%%%%
%%%%%%%%%%%%%%%%%%%%%%%%%%%%%%%%%%%%%%%%%%%%%%%%%%%%%%%%%%%%%%%%%

\section{Towards a classification of dispersive deformations}\label{section:classification}

In this section we consider the problem of classification of dispersive integrable deformations of principal hierarchies for flat F-manifolds and observe the central role played in it by the DR hierarchies. We propose two a priori different classes of deformations and we classify them, up to some finite order in $\eps$, for $1$ and $2$ dimensional flat F-manifolds, respectively. Up to that approximation, we observe that both classes contain essentially the DR hierarchies considered in Section \ref{section:principal hierarchy and dispersive deformations}.\\

\subsection{Dispersive deformations of DR type and the rank $1$ case}

\subsubsection{Integrable systems of DR type}\label{section:classification of rank 1 hierarchies of DR type}

Given a local vector field $\oX \in (\hLambda^1)^{[1]}$, consider the operator $\mathcal{D}_{\oX}\colon\hcA^1[[z]]\to\hcA^1[[z]]$ defined by
\begin{align*}
& \mathcal{D}_{\oX} Y(z):= \d_x(D-1) Y(z) - z[\oX,Y(z)],\\
& Y(z)=\sum_{k\geq 0}Y_{k-1} z^k, \qquad Y_{k-1} \in \hcA^1.
\end{align*}
Suppose there exist $N$ solutions $Y_\alpha(z) \in (\hcA^1)^{[1]}$, $1\leq \alpha\leq N$, to the equation $ \mathcal{D}_{\oX} Y_\alpha(z) = 0$ with the initial conditions $Y_\alpha(z=0) = - \theta_{\alpha,1}$. Then a new vector of solutions with the same initial conditions can be found by the following transformation:
\begin{equation}\label{eq:transformation}
Y_\alpha(z)\mapsto a^\mu_\alpha(z) Y_\mu(z),
\end{equation}
where $a_\alpha^\mu(z) = \delta_\alpha^\mu + \sum_{i>0} a^\mu_{\alpha,i} z^i \in \mbC[[z]]$.

\medskip

\begin{theorem}\label{theorem:DR type}
Assume that $\oX \in (\hLambda^1)^{[1]}$ satisfies the following properties:

\begin{itemize}
\item[(a)] there exist $N$ solutions $Y_\alpha(z)= \sum_{d\geq 0} Y_{\alpha,d-1}z^d \in (\hcA^1)^{[1]}[[z]]$, $1\leq \alpha\leq N$, to the equation
\begin{equation}\label{eq:DR type recursion}
\mathcal{D}_{\oX} Y_\alpha(z) = 0
\end{equation}
with the initial conditions $Y_\alpha(z=0) = - \theta_{\alpha,1}$,

\item[(b)] $\frac{\delta}{\delta u^\un}\oX= -u^\alpha \theta_{\alpha,1} + \d_x^2 R$, $R\in (\hcA^1)^{[-1]}$, where $\frac{\delta}{\delta u^\un}=A^\alpha\frac{\delta}{\delta u^\alpha}$ and $A^\alpha$ are some complex constants.
\end{itemize}
Then, up to a transformation of type \eqref{eq:transformation}, we have
\begin{itemize}
\item[(i)] $\displaystyle Y_{\un,0} =-u^\alpha \theta_{\alpha,1} + \d_x^2  (D-1)^{-1} R$,

\item[(ii)] $\displaystyle \oY_{\un,1} = \oX$,

\item[(iii)] $\displaystyle \left[\oY_{\alpha_1,d_1},\oY_{\alpha_2,d_2}\right] = 0$, $1\leq\alpha_1,\alpha_2\leq N$, $d_1,d_2\geq -1$,

\item[(iv)] $\displaystyle[\oY_{\alpha_2,0},Y_{\alpha_1,d}] = \d_x \frac{\d }{\d u^{\alpha_2}}Y_{\alpha_1,d+1}$, $1\leq\alpha_1,\alpha_2\leq N$, $d\geq -1$,

\item[(v)] $\displaystyle \frac{\d}{\d u^\un} Y_{\alpha,d+1} = Y_{\alpha,d}$, $1\leq\alpha\leq N$, $d\geq -1$.
\end{itemize}
\end{theorem}
\begin{proof}
The proof follows closely the proof of \cite[Theorem 5.1--5.2]{BDGR19} with Lie brackets of densities of local vector fields replacing Poisson brackets of differential polynomials.
\end{proof}

\smallskip

\begin{remark}\label{remark:DR type}
When we restrict to $\eps = 0$, a particular local vector field satisfying condition~(a) of Theorem~\ref{theorem:DR type} is given by $\oX=-(D-2)\int F^\alpha(u^1,\ldots,u^N) \theta_{\alpha,1}dx$ where the functions $F^\alpha(t^1,\ldots, t^N)$ are solutions to the oriented WDVV equations~\eqref{eq:axiom1 of flat F-man},~\eqref{eq:axiom2 of flat F-man}. It is easy to check that for such $\oX$ solutions $Y_\alpha(z)$ are given by $Y_{\alpha}(z)=-\sum_{d\ge -1} X_{\alpha,d}^\beta\theta_{\beta,1}z^{d+1}$ where the functions~$X^\beta_{\alpha,d}$ form a calibration of the flat F-manifold satisfying $X^\beta_{\alpha,-1}=\delta^\beta_\alpha$ (see Section~\ref{subsection:principal hierarchy}). Therefore, the functions $Y_\alpha(z)$ are the generating series of densities of local vector fields of the principal hierarchy of the flat F-manifold. Note that condition (b) for our $\oX$ is equivalent to $\frac{\d F^\alpha}{\d t^\un}=t^\alpha$, which can always be fulfilled by adding to $F^\alpha$ appropriate linear terms. 
\end{remark}

\smallskip

\begin{definition}
Let $\oX \in (\hLambda^1)^{[1]}$ satisfy the hypothesis of Theorem \ref{theorem:DR type}. Then we say that $\oX=\oY_{\un,1}$ and the induced hierarchy of compatible densities of local vector fields $Y_{\alpha,d}$, $1\leq\alpha\leq N$, $d\geq -1$, are {\it of double ramification (DR) type}.
\end{definition}

\smallskip

\begin{theorem}
The double ramification hierarchy \eqref{eq:DR densities} associated to an F-CohFT is a hierarchy of double ramification type.
\end{theorem}
\begin{proof}
Hypotheses~(a) and~(b) of Theorem \ref{theorem:DR type} follow immediately from claims~(iii) and~(v), respectively, of Theorem~\ref{theorem:Main}.
\end{proof}

\medskip

\subsubsection{Classification of rank $1$ hierarchies of DR type}

Thanks to Theorem~\ref{theorem:DR type} and Remark~\ref{remark:DR type}, it makes sense to use equation~\eqref{eq:DR type recursion} to find all possible deformations of DR type of a principal hierarchy associated to a given flat F-manifold, at low order in the dispersion parameter~$\eps$. These deformations will, in particular, include the ones coming from all F-CohFTs with the given genus $0$ part.\\

Consider the ancestor principal hierarchy associated to the genus $0$ part of the trivial CohFT, i.e., the CohFT with $V = \mbC\langle e \rangle$ and $c_{g,n}(e^{\otimes n}) = 1 \in H^0(\oM_{g,n})$ for all $(g,n)$ in the stable range. Let $e_1=e$, $u_k:=u^1_k$, $\theta_k:=\theta_{1,k}$ for $k\geq 0$, with $u:=u_0$, $\theta:=\theta_1$ as usual, and $Y_{d}:=Y_{1,d} = Y_{\un,d}$. A direct computation (at the approximation up to $\eps^9$) shows that its most general deformation of DR type is either of the form
\begin{equation}\label{eq:dKdV deformation}
\begin{split}
\frac{\d u}{\d t_1}=& \frac{\delta \oY_{1}}{\delta \theta}= u u_1+\eps ^2 C_{1,1}u_3 +\eps ^4 \left(C_{2,1}u_5 +{ C_{2,2}}u_2 u_3 \right)\\
&+\eps ^6 \left( C_{3,1}u_2 u_5 +{ C_{3,2}} u_3\left(u_2\right){}^2+ \left(\frac{10 \left(C_{2,1}\right){}^2}{7 C_{1,1}}+\frac{9}{35} C_{1,1} { C_{2,2}}\right)u_7\right.\\
&\left.\hspace{0.8cm}+
    \left(-\frac{8 C_{2,1} { C_{2,2}}}{3 C_{1,1}}+3 C_{3,1}\right)u_3 u_4\right)\\
&+\eps ^8 \left( C_{4,1}\left(u_2\right){}^2 u_5 +{ C_{4,2}}\left(u_2\right){}^3 u_3 +\left(\frac{15 \left(C_{2,1}\right){}^3}{7
   \left(C_{1,1}\right){}^2}+\frac{58}{105} C_{2,1} { C_{2,2}}+\frac{4}{35} C_{1,1} C_{3,1}\right)u_9\right.\\
&\left.\hspace{0.8cm}  + \left(-\frac{123 \left(C_{2,1}\right){}^2 { C_{2,2}}}{28 \left(C_{1,1}\right){}^2}+\frac{57}{100}
   \left({ C_{2,2}}\right){}^2+\frac{9}{16} C_{1,1} { C_{3,2}}+\frac{33 C_{2,1} C_{3,1}}{7 C_{1,1}}\right)u_2 u_7\right.\\
&\left.\hspace{0.8cm}   + \left(-\frac{177 \left(C_{2,1}\right){}^2 { C_{2,2}}}{4 \left(C_{1,1}\right){}^2}+\frac{201}{100}
   \left({ C_{2,2}}\right){}^2+\frac{333}{80} C_{1,1} { C_{3,2}}+\frac{33 C_{2,1} C_{3,1}}{C_{1,1}}\right)u_4 u_5\right.\\
&\left.\hspace{0.8cm}   + \left(\frac{44 C_{2,1}
   \left({ C_{2,2}}\right){}^2}{21 \left(C_{1,1}\right){}^2}-\frac{55 C_{2,1} { C_{3,2}}}{12 C_{1,1}}-\frac{44 { C_{2,2}} C_{3,1}}{21 C_{1,1}}+2C_{4,1}\right)\left(u_3\right){}^3\right.\\
&\left.\hspace{0.8cm}  + \left(-\frac{24 \left(C_{2,1}\right){}^2
   { C_{2,2}}}{\left(C_{1,1}\right){}^2}+\frac{249}{175} \left({ C_{2,2}}\right){}^2+\frac{9}{4} C_{1,1} { C_{3,2}}+\frac{132 C_{2,1} C_{3,1}}{7 C_{1,1}}\right)u_3 u_6\right.
\end{split}
\end{equation}
\begin{equation*}
\begin{split}
&\left.\hspace{0.8cm}   + \left(\frac{88 C_{2,1} \left({ C_{2,2}}\right){}^2}{21
   \left(C_{1,1}\right){}^2}-\frac{55 C_{2,1} { C_{3,2}}}{6 C_{1,1}}-\frac{88 { C_{2,2}} C_{3,1}}{21 C_{1,1}}+6 C_{4,1}\right)u_2u_3 u_4\right)\\
&   +O(\eps^{10}),
\end{split}
\end{equation*}
with $C_{i,j}\in \mbC$ and $C_{1,1}\neq 0$, or of the form
\begin{equation}\label{eq:Burgers}
\frac{\d u}{\d t_1}=\frac{\delta \oY_{1}}{\delta \theta}= u u_1+\eps C u_2, \qquad C\in \mbC.
\end{equation}

\bigskip

Notice that, imposing $C_{k,2}=0$ for all $k\geq 1$ in equation \eqref{eq:dKdV deformation}, we recover the most general Hamiltonian deformation of DR type, obtained in \cite{BDGR19}, which is in turn in one to one correspondence with the most general rank $1$ CohFT. This shows that the extra parameters~$C_{k,2}$, $k\geq 1$ control the strictly non-Hamiltonian deformations (at least with respect to the Hamiltonian operator $\d_x$). We expect these to correspond to F-CohFTs that are not CohFTs.

\medskip

\begin{remark}
It is easy to check that the r.h.s of equation \eqref{eq:dKdV deformation} is a total $x$-derivative. Comparing with the results of \cite{ALM15} we see  that a similar  result can be obtained starting from generic scalar conservation laws of the form
\begin{align}
&\frac{\partial u}{\partial t_d}=\partial_x P_{d},\qquad d\ge 0,\label{eq:general scalar deformation}\\
&P_{d}=\sum_{l\geq 0}\eps^{2l}P_{d,l},\qquad P_{d,l}\in\cA^{[2l]}_M,\notag
\end{align}
choosing
\begin{itemize}
\item $P_{d,0}=\f{u^{d+1}}{(d+1)!}$,
\item $\frac{\partial P_{1}}{\d u_{x}}=0$; the reduction to this form by means of a Miura transformation is always possible and it is unique,
\end{itemize}
and imposing the following conditions:
\begin{itemize}
\item {\it Commutativity of the flows}.
\item {\it String property}: $\frac{\d}{\d u} P_{d+1} = P_{d}$ for $d\ge -1$, where $P_{-1}:=1$.
\end{itemize}
According to the conjecture formulated  in \cite{ALM15}, it should be possible to write all the coefficients  appearing in the deformation as functions of the coefficients  of the quasilinear part. Moreover, the coefficients of the quasilinear part should be constant (due to the  string property) and arbitrary. This is consistent with the formula \eqref{eq:dKdV deformation} since the additional free parameters appearing at the order $\eps^8$ are related to the coefficients of the quasilinear part by constraints obtained considering higher order conditions.
\end{remark}

\medskip

Even more intriguing is the isolated deformation \eqref{eq:Burgers}, which, up to reabsorbing the constant $C$ into the factor $\eps$, is the celebrated Burgers equation, which is dissipative and hence non-Hamiltonian. The appearence of terms with odd powers of $\eps$ in a hierarchy of DR type rules out the possibility that it is the double ramification hierarchy of an F-CohFT. However, considering that flat F-manifolds are known to appear in genus $0$ open Gromov--Witten and Saito theory \cite{PST14,BCT18,BCT19,BB19} it is tempting to conjecture that Burgers equation~\eqref{eq:Burgers} and its higher symmetries might control some version of F-CohFT on the space of Riemann surfaces with boundaries, where curves can indeed possess half-integer genus accounting for odd powers of the genus parameter $\eps$.\\

The fact that Burgers equation \eqref{eq:Burgers} and its higher symmetries form a hierarchy of DR type can be proved rigorously at all order in $\eps$ as follows.

\medskip

\begin{theorem}
The vector field $\oX=\int (uu_x+\eps u_{xx})\theta dx$ of the Burgers equation defines a hierarchy of DR type, i.e., it satisfies conditions~(a) and~(b) of Theorem~\ref{theorem:DR type}.
\end{theorem}
\begin{proof}
Let us first present a reformulation of the Schouten--Nijenhuis bracket $[\cdot,\cdot]\colon\hLambda^1\times \hcA^1\to\hcA^1$ in terms of formal differential operators. Consider an arbitrary local vector field $\oX=\int X\theta dx\in\hLambda^1$ and a density $Y=\sum_{k\ge 0}Y_k\theta_k\in \hcA^1$. The local vector field $\oX$ defines a flow on the space of differential polynomials by
$$
\frac{\d u}{\d t}=\frac{\delta\oX}{\delta\theta}=X,
$$
and we consider also formal differential operators $\tL_{\oX}$ and $L_Y$ defined by
$$
\tL_{\oX}:=\sum_{k\ge 0}(-\d_x)^k\circ\frac{\d X}{\d u_k},\qquad L_Y:=\sum_{k\ge 0}Y_k\d_x^k.
$$
Directly from the definition~\eqref{eq:SN bracket lift 1}, we obtain the following identity:
$$
L_{[\oX,Y]}=\frac{\d}{\d t}L_Y-L_Y\circ\tL_{\oX},
$$
where we apply the differentiation $\frac{\d}{\d t}$ to the operator $L_Y$ coefficient-wise.

\bigskip

Let us now take $\oX=\int (uu_x+\eps u_{xx})\theta dx$.\\

Let us prove condition~(a) of Theorem~\ref{theorem:DR type} by showing that a required solution $Y(z)=\sum_{k\ge -1}Y_k z^{k+1}$ of equation~\eqref{eq:DR type recursion} is given by
$$
Y(z)=-e^{z\eps\d_x}e^{z(u-2\eps\d_x)}\theta_1\quad\Leftrightarrow\quad L_{Y(z)}=\sum_{k\ge -1}z^{k+1}L_{Y_k}=-e^{z\eps\d_x}\circ e^{z(u-2\eps\d_x)}\circ\d_x.
$$
Since $\tL_{\oX}=-u\d_x+\eps\d_x^2$, equation~\eqref{eq:DR type recursion} is equivalent to
\begin{gather}\label{eq:DR recursion for Burgers,1}
\d_x\circ\tD L_{Y(z)}=z\left(\frac{\d L_{Y(z)}}{\d t}-L_{Y(z)}\circ (-u\d_x+\eps\d_x^2)\right),
\end{gather}
where $\tD:=\sum_{n\ge 0}u_n\frac{\d}{\d u_n}+\eps\frac{\d}{\d\eps}$, and we apply $\tD$ to $L_{Y(z)}$ coefficient-wise. Note that $\tD L_{Y(z)}=z\frac{\d}{\d z}L_{Y(z)}$. Therefore, equation~\eqref{eq:DR recursion for Burgers,1} is equivalent to
\begin{align*}
-\d_x\circ\frac{\d}{\d z}\left(e^{z\eps\d_x}\circ e^{z(u-2\eps\d_x)}\circ\d_x\right)=&-\frac{\d}{\d t}\left(e^{z\eps\d_x}\circ e^{z(u-2\eps\d_x)}\circ\d_x\right)\\
&+e^{z\eps\d_x}\circ e^{z(u-2\eps\d_x)}\circ\d_x\circ (-u\d_x+\eps\d_x^2)\Leftrightarrow\\
\Leftrightarrow e^{z\eps\d_x}\circ(-\d_x\circ u+\eps\d_x^2)\circ e^{z(u-2\eps\d_x)}\circ\d_x=&-e^{z\eps\d_x}\circ\frac{\d}{\d t}e^{z(u-2\eps\d_x)}\circ\d_x\\
&+e^{z\eps\d_x}\circ e^{z(u-2\eps\d_x)}\circ (-\d_x\circ u+\eps\d_x^2)\circ\d_x\Leftrightarrow\\
\Leftrightarrow \frac{\d}{\d t}e^{z(u-2\eps\d_x)}=&[e^{z(u-2\eps\d_x)},-\d_x\circ u+\eps\d_x^2].
\end{align*}
Note that the last equation follows from the elementary identity $\frac{\d}{\d t}(u-2\eps\d_x)=[u-2\eps\d_x,-\d_x\circ u+\eps\d_x^2]$.\\

Condition~(b) of Theorem~\ref{theorem:DR type} immediately follows from the equation $\frac{\delta}{\delta u}\oX=-u\theta_1+\eps\theta_2$.\\
\end{proof}

\vspace{4cm}

%%%%%%%%%%%%%%%%%%%%%%%%%%%%%%%%%%%%%%%%%%%%%%%

\subsection{Homogeneous dispersive deformations and the rank $2$ case}\label{biflatdef}
\subsubsection{Homogeneous deformations with string and dilaton property}\label{subsubsection:homogeneous string dilaton}
Let us fix a flat F-manifold together with a homogeneous calibration of standard type. We consider systems of evolutionary PDEs of the form
\begin{align}
&\frac{\partial u^{\alpha}}{\partial t^\beta_d}=\partial_x P^\alpha_{\beta,d},\qquad 1\le \alpha,\beta\le N,\quad d\ge 0,\label{eq:general deformation}\\
&P^\alpha_{\beta,d}=\sum_{l\geq 0}\eps^{2l}P^\alpha_{\beta,d,l},\qquad P^\alpha_{\beta,d,l}\in\cA^{[2l]}_M,\notag
\end{align}
such that the following properties are satisfied:
\begin{enumerate}
\item {\it Commutativity of the flows}: the flows $\frac{\d}{\d t^\beta_d}$ pairwise commute,
\item The dispersionless limit of the system~\eqref{eq:general deformation} coincides with the principal hierarchy of the given calibrated flat F-manifold,
\item {\it String property}: $\frac{\d}{\d u^\un} P^{\alpha}_{\beta,d+1} = P^{\alpha}_{\beta,d}$ for $d\ge -1$, where $P^\alpha_{\beta,-1}:=\delta^\alpha_\beta$,
\item  {\it Dilaton property}: $\frac{\partial P^{\alpha}_{\un,1}}{\partial u^\beta}=D P^{\alpha}_{\beta,0}$,
\item {\it Homogeneity condition}: $\hE_\gamma P(z)=z\frac{\d}{\d z}P(z)+[P(z),Q]+P(z)\tR(z)$ for some $\gamma\in\mbC$, where $P(z):=\sum_{d\ge -1}(P^\alpha_{\beta,d})z^{d+1}$.
\end{enumerate}

\bigskip

In this section, working out the $N=2$ case, we observe how descendant DR hierarchies appear in the problem of classification of dispersive integrable deformations of principal hierarchies of flat F-manifolds of the above form, which we refer to as a {\it homogeneous deformation with string and dilaton properties}. The role played by conditions (3), (4), and (5) is central in producing finite dimensional spaces of deformations even without having to quotient with respect to equivalence up to Miura transformations of the dependent variables.

\medskip

\begin{remark}
Axioms (1), (3), and (4) above correspond closely to properties (iii), (iv), and (v) of Theorem \ref{theorem:DR type} for hierarchies of local vector fields of DR type. Homogeneity (5) corresponds to property (iii) of Proposition \ref{proposition:homogeneous DR hierarchy} for homogeneous DR hierarchies. Finally, condition (2) above is satisfied by hierarchies of DR type, see Remark \ref{remark:DR type}. This means that homogeneous dispersive deformations with string and dilaton properties contain homogeneous descendant hierarchies of DR type whose local vector fields have only even powers of $\eps$. It's not a priori clear that the converse is true and it would be interesting to investigate this point.
\end{remark}
%Let us remark that when the principal hierarchy is associated to a bi-flat F-manifold,  the hierarchy \eqref{eq: principal eq1} is equipped with an additional recurrence relation, which is called the twisted Lenard-Magri relation (see \cite{AL13bis}).

\medskip 

\subsubsection{Classification of semisimple homogeneous flat F-manifolds in dimension $2$}\label{subsubsection:2D F-manifolds classification}

In the semisimple case, using canonical coordinates $u_1,...,u_N$, the structure of a homogeneous flat F-manifold can be recovered from a solution of the following system (\cite{AL19}):
\begin{align}
&\frac{\d\Gamma^i_{ij}}{\d u_k}=-\Gamma^i_{ij}\Gamma^i_{ik}+\Gamma^i_{ij}\Gamma^j_{jk}+\Gamma^i_{ik}\Gamma^k_{kj}, && i\ne k\ne j\ne i,\notag\\
&\sum_{k=1}^N\frac{\d\Gamma^i_{ij}}{\d u_k}=0,&& i\ne j,\label{BF2}\\
&\sum_{k=1}^N u_k\frac{\d\Gamma^i_{ij}}{\d u_k}=-\Gamma^i_{ij},&& i\ne j.\label{BF3}
\end{align}

\bigskip

For $N=2$, the above system reduces to \eqref{BF2} and \eqref{BF3}, and the general solution is 
\[
\Gamma^i_{ij}=\frac{\epsilon_j}{u_i-u_j},
\]
where $\epsilon_1$ and $\epsilon_2$ are arbitrary constants. Note that the corresponding vector of conformal dimensions is equal to $(2\epsilon_2,2\epsilon_1)$. In order to compute a vector potential, we need to introduce flat coordinates $u,v$ (these correspond to $t^1, t^2$ in Section~\ref{subsection:flat F-manifolds}). We have to distinguish~$3$~cases:

\medskip

\begin{enumerate}

\item[I.] $\epsilon_1+\epsilon_2\ne 0,1$. In this case, flat coordinates are
\begin{gather*}
u=\left(\f{u_1-u_2}{4}\right)^{\f{1}{m}},\qquad v=\f{2+c}{4}u_1+\f{2-c}{4}u_2,
\end{gather*}
where $c=2\,\frac{\epsilon_1-\epsilon_2}{\epsilon_1+\epsilon_2}$, $m = \frac{1}{1-\epsilon_1-\epsilon_2}\ne 0,1$, and a vector potential is
\begin{gather*}
(F^1,F^2)=
\begin{cases}
\left(uv-2c\frac{u^{m+1}}{m+1},\frac{v^2}{2}+\frac{4-c^2}{2}\frac{mu^{2m}}{2m-1}\right),&\text{if $m\ne -1,\frac{1}{2},0,1$},\\
\left(uv-2c\log u,\frac{v^2}{2}+\frac{4-c^2}{6}u^{-2}\right),&\text{if $m=-1$},\\
\left(uv-\frac{4}{3}cu^{3/2},\frac{v^2}{2}+\frac{4-c^2}{4}u\log u\right),&\text{if $m=\frac{1}{2}$}.
\end{cases}
\end{gather*}
The unit is $\frac{\d}{\d v}$, the Euler vector field is $E=\frac{1}{m}u\frac{\d}{\d u}+v\frac{\d}{\d v}$, and $\ogamma=\left(\frac{(2-c)(m-1)}{2m},\frac{(2+c)(m-1)}{2m}\right)$.\\

If $m$ is a half-integer, these are the vector potentials of the bi-flat F-manifold structures defined on the orbit space of the dihedral group $I_2(2m)$ \cite{AL17}. If also $c=0$, the above vector potential comes from the Dubrovin--Frobenius manifold structure defined on the orbit space of the dihedral group.

\medskip

\item[II.] $\epsilon_1=c,\,\epsilon_2=1-c$, $c\ne 0$ (see the remark about the case $c=0$ below). Using the flat coordinates 
\begin{gather*}
u=u_1-u_2+\f{u_2}{c},\qquad v=-\ln{(u_1-u_2)},
\end{gather*}
we obtain 
\begin{gather}\label{eq:vector potential, case II}
F^1=\f{c}{2} u^2+\f{1-c}{4}e^{-2v},\qquad F^2=cuv+(2c-1)e^{-v}.
\end{gather}
The unit is $\frac{1}{c}\frac{\d}{\d u}$, the Euler vector field is $E=u\frac{\d}{\d u}-\frac{\d}{\d v}$, and $\ogamma=(2-2c,2c)$. For $c=\f{1}{2}$, the above vector potential comes from the genus $0$ Gromov--Witten potential of the complex projective line.\\

In the case $c=0$, choosing the flat coordinates $u=u_2$ and $v=-\ln{(u_1-u_2)}$, we obtain $F^1=\frac{u^2}{2}$ and $F^2=uv-e^{-v}$. This flat F-manifold is isomorphic to the flat F-manifold~\eqref{eq:vector potential, case II} with $c=1$ and shifted by $v\mapsto v+\pi i$.

\medskip

\item[III.] $\epsilon_1=c,\epsilon_2=-c$. If $c\ne 0$, then using the flat coordinates 
\begin{gather*}
u=u_1-u_2,\qquad v=(u_1-u_2)\ln{(u_1-u_2)}+\f{u_2}{c},
\end{gather*}
we obtain 
\begin{align*}
F^1=&cuv+u^2\left(\frac{c+1}{2}-c\ln u\right),\\
F^2=&\f{c}{2}v^2+u^2\left(-\frac{3c+1}{4}+\frac{c+1}{2}\ln u-\frac{c}{2}(\ln u)^2\right).
\end{align*}
The unit is $\frac{1}{c}\frac{\d}{\d v}$, the Euler vector field is $E=u\frac{\d}{\d u}+(u+v)\frac{\d}{\d v}$, and $\ogamma=(-2c,2c)$.\\

If $c=0$, then choosing as flat coordinates the canonical coordinates $u=u_1$ and $v=u_2$ we obtain
$$
F^1=\frac{u^2}{2},\qquad F^2=\frac{v^2}{2}.
$$
The unit is $\frac{\d}{\d u}+\frac{\d}{\d v}$, the Euler vector field is $E=u\frac{\d}{\d u}+v\frac{\d}{\d v}$, and $\ogamma=(0,0)$.
\end{enumerate}

\bigskip

\subsubsection{Integrable deformations of rank $2$ homogeneous principal hierarchies}\label{subsubsection:integrable deformations of principal hierarchy}

We now want to classify all homogeneous deformations with string and dilaton properties of principal hierarchies associated to the homogeneous two-dimensional flat F-manifolds considered above. In our computations below, we have observed the following remarkable facts:\\

\begin{itemize}
\item If such a deformation exists and is nontrivial at the $\eps^2$ approximation, then $\gamma$ {\bf must be equal} to $\gamma_1$ or $\gamma_2$.\\

\item For $\gamma=\gamma_i$, at the $\eps^2$ approximation, any such deformation coincides with the descendant DR hierarchy constructed using an appropriate framing. In particular, any such deformation at the approximation up to $\eps^2$ can be extended to a deformation at all orders of $\eps$.
\end{itemize}

\bigskip

Let us consider all three cases from Section \ref{subsubsection:2D F-manifolds classification} in detail.\\

{\it \underline{Case I}}. For simplicity, we consider the case $\frac{1}{m}\ne\mbZ$, which guarantees that there is a unique homogeneous calibration of standard type such that $\tR_i=0$ for $i\ge 1$. Recall that the vector of conformal dimensions is $(\gamma_1,\gamma_2)=\left(\frac{(2-c)(m-1)}{2m},\frac{(2+c)(m-1)}{2m}\right)$. We have three subcases.\\

{\it Case I1}. If $\gamma_1\ne\gamma_2$ and $\gamma=\gamma_1$, we obtain
\begin{align*}
P^1_{1,0}=&v-2cu^m\\
&+Au^{-\f{1}{2}c(m-1)}\eps^2\left(m(c-2)(cm-c-2m+4)u^{m-3}u_x^2+m(c-2)^2u^{m-2}u_{xx}-cu^{-1}v_{xx}\right)+\mathcal{O}(\eps^4),\\
P^2_{1,0}=&\f{m^2(4-c^2)}{2m-1}u^{2m-1}+Am(c^2-4)u^{-\f{1}{2}c(m-1)}\eps^2\left(m(cm-c-4m+6)u^{2m-4}u_x^2\right.\\
&\left.+m(c-4)u^{2m-3}u_{xx}-u^{m-2}v_{xx}\right)+\mathcal{O}(\eps^4),
\end{align*} 
and
\begin{align*}
P_{2,1}^1=&uv-\f{2mc}{m+1}u^{m+1}+Au^{-\f{1}{2}c(m-1)}\eps^2\bigg(m(c-2)(cm-c-2m-2)u^{m-2}u_x^2\\
&+\f{m(c-2)(cm-c-2m-2)}{m-1}u^{m-1}u_{xx}-\f{cm-c-4}{m-1}v_{xx}\bigg)+\mathcal{O}(\eps^4),\\
P_{2,1}^{2}=&\f{v^2}{2}+\f{m(4-c^2)}{2}u^{2m}+Am(c+2)u^{-\f{1}{2}c(m-1)}\eps^2\bigg(m(cm-c-4m)(c-2)u^{2m-3}u_x^2\\
&+\f{m(cm-c-4m)(c-2)}{m-1}u^{2m-2}u_{xx}-\f{cm-c-2m-2}{m-1}u^{m-1}v_{xx}\bigg)+\mathcal{O}(\eps^4).
\end{align*}
Here $A$ is an arbitrary complex constant. This deformation is given by the descendant DR hierarchy corresponding to the framing $(\cX^1,\cX^2)=12A\left(\frac{4}{m-1}u^{-\frac{1}{2}c(m-1)},\frac{4m(c+2)}{m-1}u^{-\frac{1}{2}(c-2)(m-1)}\right)$.\\

{\it Case I2}. If $\gamma_1\ne\gamma_2$ and $\gamma=\gamma_2$, we obtain
\begin{align*}
P^1_{1,0}=&v-2cu^m\\
&+Bu^{\f{1}{2}c(m-1)}\eps^2\left(m(c+2)(cm-c+2m-4)u^{m-3}u_x^2+m(c+2)^2u^{m-2}u_{xx}-cu^{-1}v_{xx}\right)+\mathcal{O}(\eps^4),\\
P^2_{1,0}=& \f{m^2(4-c^2)}{2m-1}u^{2m-1}+Bm(c^2-4)u^{\f{1}{2}c(m-1)}\eps^2\left(m(cm-c+4m-6)u^{2m-4}u_x^2\right.\\
&\left.+m(c+4)u^{2m-3}u_{xx}-u^{m-2}v_{xx}\right)+\mathcal{O}(\eps^4),
\end{align*}
and  
\begin{align*}
P_{2,1}^{1}=&uv-\f{2mc}{m+1}u^{m+1}+Bu^{\f{1}{2}c(m-1)}\eps^2\bigg(m(c+2)(cm-c+2m+2)u^{m-2}u_x^2
\\
&+\f{m(c+2)(cm-c+2m+2)}{m-1}u^{m-1}u_{xx}-\f{cm-c+4}{m-1}v_{xx}\bigg)+\mathcal{O}(\eps^4),\\
P_{2,1}^{2}=&\f{v^2}{2}+\f{m(4-c^2)}{2}u^{2m}+Bm(c-2)u^{\f{1}{2}c(m-1)}\eps^2\bigg(m(cm-c+4m)(c+2)u^{2m-3}u_x^2\\
&+\f{m(cm-c+4m)(c+2)}{m-1}u^{2m-2}u_{xx}-\f{cm-c+2m+2}{m-1}u^{m-1}v_{xx}\bigg)+\mathcal{O}(\eps^4).
\end{align*}
Here $B$ is an arbitrary complex constant. This deformation is given by the descendant DR hierarchy corresponding to the framing $(\cX^1,\cX^2)=12B\left(-\frac{4}{m-1}u^{\frac{1}{2}c(m-1)},-\frac{4m(c-2)}{m-1}u^{\frac{1}{2}(c+2)(m-1)}\right)$.\\

{\it Case I3}. If $\gamma_1=\gamma_2$ (which is equivalent to $c=0$) and $\gamma$ coincides with them, we get a two-parameter family of deformations formed by linear combinations of the deformations from Cases~I1 and~I2.\\

{\it\underline{Case II}}. There is a unique homogeneous calibration of standard type such that $X^\alpha_{\beta,0}=\frac{\d F^\alpha}{\d t^\beta}$, $\tR_1=\begin{pmatrix}0 & 0\\ -c & 0 \end{pmatrix}$, and $\tR_i=0$ for $i\ge 2$. Recall that the vector of conformal dimensions is $(\gamma_1,\gamma_2)=(2-2c,2c)$. We have three subcases.\\

{\it Case II1}. If $\gamma_1\ne\gamma_2$ and $\gamma=\gamma_1$, we obtain
\begin{align*}
P^1_{2,0}=&\f{c-1}{2}e^{-2v}+A(c-1)e^{2(c-1)v}\eps^2\left(u_{xx}-\f{2c-3}{2c}e^{-v}v_x^2+\frac{2c-3}{c}e^{-v}v_{xx}\right)+\mathcal{O}(\eps^4),\\
P^2_{2,0}=&cu-(2c-1)e^{-v}+A e^{2(c-1)v}\eps^2\left(-\f{2c-1}{2}e^v u_{xx}+\frac{(c-1)^2}{c}v_x^2-\frac{(c-1)^2}{c} v_{xx}\right)+\mathcal{O}(\eps^4),
\end{align*} 
and
\begin{align*}
P_{1,1}^{1}=&\f{c^2}{2}u^2+\f{c(c-1)}{4}(2v+1)e^{-2v}+Ae^{2(c-1)v}\eps^2\bigg(c((c-1)v+1)u_{xx}\\
&-\f{c-1}{2}((2c-3)v+3)e^{-v}v_x^2+\f{c-1}{2}((2c-3)v+2)e^{-v}v_{xx}\bigg)+\mathcal{O}(\eps^4),\\
P_{1,1}^{2} =&c^2uv-c(2c-1)(v+1)e^{-v}+Ae^{2(c-1)v}\eps^2\left(-\f{c}{2}((2c-1)v+2)e^{v}u_{xx}\right.\\
&\left.+\f{c-1}{2}((2c-2)v+3)v_x^2-(c-1)((c-1)v+1)v_{xx}\right)+\mathcal{O}(\eps^4),
\end{align*}
where $A$ is an arbitrary complex constant. This deformation is given by the descendant DR hierarchy corresponding to the framing $(\cX^1,\cX^2)=\frac{12A}{c}\left(e^{2(c-1)v},-e^{(2c-1)v}\right)$.\\

{\it Case II2}. If $\gamma_1\ne\gamma_2$ and $\gamma=\gamma_2$, we obtain
\begin{align*}
P^1_{2,0}=&\f{c-1}{2}e^{-2v}+B(c-1)e^{-2cv}\eps^2\left(-u_{xx}+\f{2c+1}{2c}e^{-v}v_x^2
-\f{2c+1}{2c}e^{-v}v_{xx}\right)+\mathcal{O}(\eps^4),\\
P^2_{2,0}=&cu-e^{-v}(2c-1)+Be^{-2cv}\eps^2\left(\frac{2c-1}{2}e^{v}u_{xx}-c v_x^2+c v_{xx}\right)+\mathcal{O}(\eps^4),
\end{align*} 
and 
\begin{align*}
P_{1,1}^{1}=&\f{c^2}{2}u^2+\f{c(c-1)}{4}(2v+1)e^{-2v}+B(c-1)e^{-2cv}\eps^2\left(-(cv-1)u_{xx}+\f{1}{2}((2c+1)v-3)e^{-v}v_x^2\right.\\
&\left.-\f{1}{2}((2c+1)v-2)e^{-v}v_{xx}\right)+\mathcal{O}(\eps^4),\\
P_{1,1}^{2}=&c^2uv-c(2c-1)(v+1)e^{-v}\\
&+Bce^{-2cv}\eps^2\left(\f{1}{2}((2c-1)v-2)e^{v}u_{xx}-\f{1}{2}(2cv-3)v_x^2+(cv-1)v_{xx}\right)+\mathcal{O}(\eps^4),
\end{align*}
where $B$ is an arbitrary complex constant. This deformation is given by the descendant DR hierarchy corresponding to the framing $(\cX^1,\cX^2)=\frac{12B}{c}\left(\frac{c-1}{c}e^{-2cv},-e^{-(2c-1)v}\right)$.\\

{\it Case II3}. If $\gamma_1=\gamma_2$ (which is equivalent to $c=\frac{1}{2}$) and $\gamma$ coincides with them, we get a two-parameter family of deformations formed by linear combinations of the deformations from Cases~II1 and~II2.\\

{\it\underline{Case III}}. There is a unique homogeneous calibration of standard type such that $\tR_i=0$ for $i\ge 1$. Recall that the vector of conformal dimensions is $(\gamma_1,\gamma_2)=(-2c,2c)$.\\

{\it Case III1}. If $\gamma_1\ne\gamma_2$ (equivalently, $c\ne 0$) and $\gamma=\gamma_1=-2c$, we obtain 
\begin{align*}
P^1_{1,0}=&u(1-2c\ln u)+cv+Au^{-2c-1}\eps^2\left(-\left(c+\f{1}{2}\right)u^{-1}u_x^2+ cv_{xx}+\left(\f{3}{2}-c(1+\ln u)\right)u_{xx}\right)+\mathcal{O}(\eps^4),\\
P^2_{1,0}=&u(\ln u-c(1+\ln^2 u))+Au^{-2c-1}\eps^2\left(-u^{-1}\left(\left(c+\f{1}{2}\right)\ln u+c\right)u_x^2+\right.\\
&\left.+\left(c(1+\ln u)-\f{1}{2}\right)v_{xx}+\left(2(1+\ln u)-c(1+\ln u)^2-\f{1}{2c}\right)u_{xx}\right)+\mathcal{O}(\eps^4),
\end{align*} 
and
\begin{align*}
P^1_{2,1}=&-\f{c}{2}u(cu(1+2\ln u)-2cv-u)+Au^{-2c-1}\eps^2\left(c(c-1)uv_{xx}+c(1-c)u_x^2+\right. \\& \left.+u\left(c(1-c)\ln u-\left(c-\frac{1}{2}\right)(c-2)\right)u_{xx}\right)+\mathcal{O}(\eps^4),\\
P^2_{2,1}=&-\f{c}{4}u^2(2(c\ln u+c-1)+c-1)\ln u+\f{c^2}{2}v^2+\\
&+Au^{-2c-1}\eps^2\left(c\left((1-c)\ln u+\f{3}{2}-c\right)u_x^2+cu\left((c-1)\ln u+c-\f{3}{2}\right)v_{xx}  \right.+\\&\hspace{2.5cm}\left.+u\left(c(1-c)\ln^2 u-(2c^2-4c+1)\ln u-c^2+3c-\f{3}{2}\right)u_{xx}\right)+\mathcal{O}(\eps^4),
\end{align*} 
where $A$ is an arbitrary complex constant. This deformation is given by the descendant DR hierarchy corresponding to the framing $(\cX^1,\cX^2)=-\frac{12A}{c}(u^{-2c},u^{-2c}(1+\ln u))$.\\

{\it Case III2}. If $\gamma_1\ne\gamma_2$ (equivalently, $c\ne 0$) and $\gamma=\gamma_2=2c$, we obtain 
\begin{align*}
P^1_{1,0}=&u(1-2c\ln u)+cv+Bu^{2c-1}\eps^2\left(\left(c-\f{1}{2}\right)u^{-1}u_x^2-cv_{xx}+\left(c\ln u+c+\f{1}{2}\right)u_{xx}\right)+\mathcal{O}(\eps^4),\\
P^2_{1,0}=&u(\ln u-c(1+\ln^2 u))+Bu^{2c-1}\eps^2\left(\left(\f{1}{2}-c(1+\ln u)\right)v_{xx} +\right.\\&\left. +\left(c(\ln u+1)^2-\f{1}{2c}\right)u_{xx}+u^{-1}\left(\left(c-\f{1}{2}\right)\ln u+c-1+\f{1}{2c}\right)u_x^2\right)+\mathcal{O}(\eps^4),
\end{align*} 
and
\begin{align*}
P^1_{2,1}=&-\f{c}{2}u(cu(1+2\ln u)-2cv-u)+\f{c}{2}Bu^{2c-1}\eps^2\left(2(c+1)u_x^2-2(c+1)uv_{xx}+\right.\\&\left.+u(2(c+1)\ln u+2c+3)u_{xx}\right)+\mathcal{O}(\eps^4),\\
P^2_{2,1}=&-\f{c}{4}u^2(2(c\ln u+c-1)+c-1)\ln u+\f{c^2}{2}v^2+\\ &+Bu^{2c-1}\eps^2\left(\left(c^2+\f{c}{2}-1+c(c+1)\ln u\right)u_x^2+u\left(1-\f{c}{2}-c^2-c(c+1)\ln u\right)v_{xx}+\right.\\
&\hspace{2.3cm}\left.+u\left(c^2+c-\f{3}{2}+c(c+1)(\ln u+2)\ln u)\right)u_{xx}\right)+\mathcal{O}(\eps^4),
\end{align*} 
where $B$ is an arbitrary complex constant. This deformation is given by the descendant DR hierarchy corresponding to the framing  $(\cX^1,\cX^2)=\frac{12B}{c}\left(-u^{2c},u^{2c}\left(\frac{1}{c}-1-\ln u\right)\right)$.\\

{\it Case III3}. If $\gamma_1=\gamma_2$ (equivalently, $c=0$) and $\gamma=\gamma_i=0$, we get the two-parameter family of deformations
\begin{align*}
P^1_{1,1}=&\f{u^2}{2}+A\eps^2 u_{xx}+\mathcal{O}(\eps^4),\\
P^2_{1,1}=&\f{v^2}{2}+B\eps^2 v_{xx}+\mathcal{O}(\eps^4),
\end{align*}
which is given by the descendant DR hierarchy corresponding to the framing  $(\cX^1,\cX^2)=12(A,B)$. This hierarchy is just the DR hierarchy of the rank $2$ F-topological field theory
$$
c_{g,n+1}(e^{\alpha_0}\otimes\otimes_{i=1}^n e_{\alpha_i})=
\begin{cases}
A^g,&\text{if $\alpha_0=\ldots=\alpha_n=1$},\\
B^g,&\text{if $\alpha_0=\ldots=\alpha_n=2$},\\
0,&\text{otherwise},
\end{cases}
$$
and it coincides with the system of two uncoupled KdV hierarchies.\\
 
\subsection{General integrable deformations and open problems}\label{sec:final}

In Section~\ref{subsubsection:homogeneous string dilaton}, we considered the problem of classification of dispersive deformations, containing only even powers of~$\eps$ and satisfying properties~(1)--(5), of principal hierarchies of two-dimensional homogeneous semisimple flat F-manifolds. We observed that at the approximation up to $\eps^2$ all such deformations are given by the descendant DR hierarchies.\\

In this section, we consider more general dispersive deformations of the same principal hierarchies: first, we allow odd powers of $\eps$ in the dispersive deformation~\eqref{eq:general deformation}, and, second, we require that only properties (1)--(2) are satisfied. In other words, we require only integrability, i.e., pairwise commutativity of the flows. In the table below, we summarize the results of computations of such deformations at the approximation up to~$\eps^2$ (the results for Case I were already obtained in \cite{AL18}). When we refer to a functional parameter relative to an integrable deformation, we mean that at a specified order the equivalence classes of deformations depend on an arbitrary function. Recall (see Definition~\ref{definition:equivalent deformations}) that two deformations are said to be equivalent if they are related by a Miura transformation that is close to identity.\\
\begin{table}[h]
\begin{center}
\begin{tabular}{|p{30mm}|p{30mm}|p{30mm}|p{30mm}|}
\hline
{\bf Case} & {\bf Values of $c$ } & {\bf Integrable first order deformations} &
{\bf Integrable second order deformations}    \\ 
\hline
I & $c\neq \pm 2$  & Miura trivial & Two functional parameters   \\
\hline
I & $c=\pm 2$  & One functional parameter  & Two functional parameters   \\
\hline
II & $c\neq 1$  & Miura trivial & Two functional parameters   \\
\hline
II & $c=1$  & One functional parameter  & Two functional parameters   \\
\hline
III & arbitrary $c$  & Miura trivial & Two functional parameters   \\
\hline
\end{tabular}
\end{center}
\caption{Functional parameters for the integrable deformations at the approximation up to $\eps^2$}
\label{tab1}
\end{table} 

For special values of the functional parameters, we recover the genus one approximations of the descendant DR hierarchies from Section~\ref{subsubsection:integrable deformations of principal hierarchy}. Unfortunately, for generic choices of the functional parameters the existence of a full dispersive hierarchy is an open problem. Toward this direction, let us point out that in  \cite{AL18} it was conjectured that up to equivalence integrable deformations are labelled by a simple set of invariants called {\it Miura invariants}. Consider a system of evolutionary PDEs of the form
\begin{align}
\frac{\d u^\alpha}{\d t}=A^\alpha_\beta(u^*)u^\beta_x+&\eps\left(B^\alpha_\beta(u^*)u^\beta_{xx}+B^\alpha_{\beta\gamma}(u^*)u^\beta_x u^\gamma_x\right)\label{intsys2}\\
+&\eps^2\left(C^\alpha_\beta(u^*)u^\beta_{xxx}+C^\alpha_{\beta\gamma}(u^*)u^\beta_x u^\gamma_{xx}+C^\alpha_{\beta\gamma\delta}(u^*)u^\beta_x u^\gamma_x u^\delta_x\right)+\ldots,\quad \alpha=1,\ldots,N,\notag
\end{align}
and as in the proof of Theorem~\ref{theorem:nonequivalent deformations} consider the associated Miura matrix
$$
M^\alpha_\beta(u^*,p)=A^\alpha_\beta(u^*)+B^\alpha_\beta(u^*)p+C^\alpha_\beta(u^*)p^2+\ldots.
$$
The {\em Miura invariants} of the system \eqref{intsys2} are the eigenvalues $\lambda^i(u^*,p)$ of the Miura matrix. If the eigenvalues of the matrix $(A^\alpha_\beta)$ are pairwise distinct at some point $(u^1,\ldots,u^N)=(u^1_\orig,\ldots,u^N_\orig)\in\mbC^N$, then the Miura invariants are well defined as formal power series whose coefficients are functions on an open neighbourhood of $(u^1_\orig,\ldots,u^N_\orig)$:
\begin{equation}\label{MiuraInv}
\lambda^i=v^i+\lambda_1^ip+\lambda_2^ip^2+\ldots,\quad i=1,\ldots,N.
\end{equation}
The functional parameters of Table~\ref{tab1} can be identified with a part of the coefficients $\lambda^1_1,\lambda^2_1,\lambda^1_2,\lambda^2_2$ in formula \eqref{MiuraInv}. The presence of odd powers of $p$ in the expansion \eqref{MiuraInv} seems an exceptional phenomenon. In the case of special deformations satisfying all the properties (1)--(5) from Section~\ref{subsubsection:homogeneous string dilaton}, there are examples related to open Gromov--Witten theory (\cite{BR18}) and we expect that this is not a coincidence. However, this point requires further investigation.


\medskip

\begin{thebibliography}{BDGR18}

\bibitem[ABLR20]{ABLR20} A. Arsie, A. Buryak, P. Lorenzoni, P. Rossi. {\it Semisimple flat F-manifolds in higher genus}. arXiv:2001.05599.

\bibitem[AL12]{AL12} A. Arsie, P. Lorenzoni. {\it Poisson bracket on $1$-forms and evolutionary partial differential equations}. Journal of Physics A: Mathematical and Theoretical~{\bf 45} (2012), no.~47, 475208.

\bibitem[AL13a]{AL13} A. Arsie, P. Lorenzoni. {\it From the Darboux--Egorov system to bi-flat F-manifolds}. Journal of Geometry and Physics~{\bf 70} (2013), 98--116.

\bibitem[AL13b]{AL13bis} A. Arsie, P. Lorenzoni. {\it F-manifolds with eventual identities, bidifferential calculus and twisted Lenard--Magri chains}. International Mathematics Research Notices {\bf 2013} (2013), no.~17, 3931--3976.

\bibitem[ALM15]{ALM15} A. Arsie, P. Lorenzoni, A. Moro. {\it On integrable conservation laws}.  Proceedings A {\bf 471} (2015), no. 2173, 20140124.

\bibitem[AL17]{AL17} A. Arsie, P. Lorenzoni. {\it Complex reflection groups, logarithmic connections and bi-flat F-manifolds}. Letters in Mathematical Physics {\bf 107} (2017), no. 10, 1919--1961.

\bibitem[AL18]{AL18} A. Arsie, P. Lorenzoni. {\it Flat F-manifolds, Miura invariants, and integrable systems of conservation laws}. Journal of Integrable Systems~{\bf 3} (2018), no. 1, xyy004.

\bibitem[AL19]{AL19} A. Arsie, P. Lorenzoni. {\it F-manifolds, multi-flat structures and Painlev\'e transcendents}. Asian Journal of Mathematics~{\bf 23} (2019), no. 5, 877--904.

\bibitem[BB19]{BB19} A. Basalaev, A. Buryak. {\it Open WDVV equations and Virasoro constraints}. Arnold Mathematical Journal~{\bf 5} (2019), no. 2--3, 145--186.

\bibitem[Bur15]{Bur15} A. Buryak. {\it Double ramification cycles and integrable hierarchies}. Communications in Mathematical Physics~{\bf 336} (2015), 1085--1107.

\bibitem[BCT18]{BCT18} A. Buryak, E. Clader, R. J. Tessler. {\it Open $r$-spin theory II: The analogue of Witten's conjecture for $r$-spin disks}. arXiv:1809.02536v4.

\bibitem[BCT19]{BCT19} A. Buryak, E. Clader, R. J. Tessler. {\it Closed extended $r$-spin theory and the Gelfand--Dickey wave function}. Journal of Geometry and Physics~{\bf 137} (2019), 132--153.

\bibitem[BDGR18]{BDGR18} A. Buryak, B. Dubrovin, J. Gu\'er\'e, P. Rossi. {\it Tau-structure for the double ramification hierarchies}. Communications in Mathematical Physics~{\bf 363} (2018), no.~1, 191--260.

\bibitem[BDGR19]{BDGR19} A. Buryak, B. Dubrovin, J. Gu\'er\'e, P. Rossi. {\it Integrable systems of double ramification type}. International Mathematics Research Notices, rnz029, https://doi.org/10.1093/imrn/rnz029.

\bibitem[BGR19]{BGR19} A. Buryak, B. Dubrovin, J. Guéré. {\it DR/DZ equivalence conjecture and tautological relations}. Geometry \& Topology~{\bf 23} (2019), no.~7, 3537--3600.

\bibitem[BR16a]{BR16a} A. Buryak, P. Rossi. {\it Recursion relations for double ramification hierarchies}. Communications in Mathematical Physics~{\bf 342} (2016), 533--568.

\bibitem[BR16b]{BR16b} A. Buryak, P. Rossi. {\it Double ramification cycles and quantum integrable systems}. Letters in Mathematical Physics~{\bf 106} (2016), 289--317.

\bibitem[BR18]{BR18} A. Buryak, P. Rossi. {\it Extended $r$-spin theory in all genera and the discrete KdV hierarchy}. arXiv:1806.09825. 

\bibitem[BRS20]{BRS20} A. Buryak, P. Rossi, S. Shadrin. {\it Towards a bihamiltonian structure for the double ramification hierarchy}. arXiv:2007.00846.

\bibitem[BSSZ15]{BSSZ15} A. Buryak, S. Shadrin, L. Spitz, D. Zvonkine. {\it Integrals of $\psi$-classes over double ramification cycles}. American Journal of Mathematics~{\bf 137} (2015), no.~3, 699--737.

\bibitem[CPS18]{CPS18} G. Carlet, H. Posthuma, S. Shadrin. {\it Deformations of semisimple Poisson pencils of hydrodynamic type are unobstructed}. Journal of Differential Geometry~{\bf 108} (2018), no.~1, 63--89.

\bibitem[DH17]{DH17} L. David, C. Hertling. {\it Regular F-manifolds: initial conditions and Frobenius metrics}. Annali della Scuola Normale di Pisa, Classe di Scienze~{\bf 17} (2017), no.~3, 1121--1152.

\bibitem[Dub96]{Dub96} B. Dubrovin. {\it Geometry of 2D topological field theories}. Integrable systems and quantum groups (Montecatini Terme, 1993), 120--348, Lecture Notes in Math., 1620, Fond. CIME/CIME Found. Subser., Springer, Berlin, 1996.

\bibitem[Dub04]{Dub04} B. Dubrovin. {\it On almost duality for Frobenius manifolds}. Geometry, topology, and mathematical physics, 75--132, Amer. Math. Soc. Transl. Ser. 2, 212, Adv. Math. Sci., 55, Amer. Math. Soc., Providence, RI, 2004.

\bibitem[DLZ06]{DLZ06} B. Dubrovin, S.-Q. Liu, Y. Zhang. {\it On Hamiltonian perturbations
of hyperbolic systems of conservation laws I: Quasi-triviality of bi-Hamiltonian perturbations}. Communications in Pure and Applied Mathematics~{\bf 59} (2006), no.~4, 559--615.

\bibitem[DZ01]{DZ01} B. Dubrovin, Y. Zhang. {\it Normal forms of hierarchies of integrable PDEs, Frobenius manifolds and Gromov--Witten invariants}. arXiv:math/0108160.

\bibitem[FP00]{FP00} C. Faber, R. Pandharipande. {\it Logarithmic series and Hodge integrals in the tautological ring. With an appendix by Don Zagier}. Michigan Mathematical Journal~{\bf 48} (2000), no.~1, 215--252.

\bibitem[Get04]{Get04} E. Getzler. {\it The jet-space of a Frobenius manifold and higher-genus Gromov--Witten invariants}. Frobenius manifolds, 45-–89, Aspects Math., E36, Friedr. Vieweg, Wiesbaden, 2004.

\bibitem[Giv01]{Giv01} A. Givental. {\it Semisimple Frobenius structures at higher genus}. International Mathematics Research Notices {\bf 2001} (2001), no.~23, 1265--1286.

\bibitem[Hai13]{Hai13} R. Hain. {\it Normal functions and the geometry of moduli spaces of curves}. Handbook of moduli. Vol.~I, 527--578, Adv. Lect. Math. (ALM), 24, Int. Press, Somerville, MA, 2013.

\bibitem[JPPZ17]{JPPZ17} F.~Janda, R. Pandharipande, A. Pixton, D. Zvonkine. {\it Double ramification cycles on the moduli spaces of curves}. Publications Math\'ematiques. Institut de Hautes \'Etudes Scientifiques~{\bf 125} (2017), 221--266. 

\bibitem[KMS15]{KMS} M. Kato, T. Mano, J. Sekiguchi. {\it Flat structure on the space of isomonodromic deformations}. arXiv:1511.01608.

\bibitem[KMS18]{KMS18} Y. Konishi, S. Minabe, Y. Shiraishi. {\it Almost duality for Saito structure and complex reflection groups}. Journal of Integrable Systems {\bf 3} (2018), no. 1, xyy003.

\bibitem[Kon92]{Kon92} M. Kontsevich. {\it Intersection theory on the moduli space of curves and the matrix Airy function}. Communications in Mathematical Physics~147 (1992), 1--23.

\bibitem[LRZ15]{LRZ15} S.-Q. Liu, Y. Ruan, Y. Zhang. {\it BCFG Drinfeld--Sokolov hierarchies and FJRW-theory}. Inventiones Mathematicae~{\bf 201} (2015), 711--772.

\bibitem[LPR09]{LPR09} P. Lorenzoni, M. Pedroni, A. Raimondo. {\it $F$-manifolds and integrable systems of hydrodynamic type}. Archivum Mathematicum~{\bf 47} (2011), no.~3, 163--180. 

\bibitem[Lor14]{Lor14} P. Lorenzoni. {\it Darboux--Egorov system, bi-flat F-manifolds and Painlev\'e VI}. International Mathematics Research Notices {\bf 2014} (2014), no. 12, 3279--3302.

\bibitem[Man05]{Man05} Y. Manin. {\it F-manifolds with flat structure and Dubrovin's duality}. Advances in Mathematics~{\bf 198} (2005), no.~1, 5--26.

\bibitem[MW13]{MW13} S. Marcus, J. Wise. {\it Stable maps to rational curves and the relative Jacobian}. arXiv:1310.5981.

\bibitem[PST14]{PST14} R. Pandharipande, J. P. Solomon, R. J. Tessler. {\it Intersection theory on moduli of disks, open KdV and Virasoro}. arXiv:1409.2191v2.

\bibitem[Ros17]{Ros17} P. Rossi. {\it Integrability, quantization and moduli spaces of curves}. Symmetry, Integrability and Geometry: Methods and Applications~{\bf 13} (2017), 060.

\bibitem[Sab98]{Sab98} C. Sabbah. {\it Frobenius manifolds: isomonodromic deformations and infinitesimal period mappings}. Expositiones Mathematicae {\bf 16} (1998), no. 1, 1--57.

\bibitem[SZ11]{SZ11} S. Shadrin, D. Zvonkine. {\it A group action on Losev--Manin cohomological field theories}. Annales de l'Institute Fourier {\bf 61} (2011), no.~7, 2719--2743.

\bibitem[Tel12]{Tel12} C. Teleman. {\it The structure of 2D semi-simple field theories}. Inventiones Mathematicae~{\bf 188} (2012), no. 3, 525--588.

\bibitem[Wit91]{Wit91} E. Witten. {\it Two-dimensional gravity and intersection theory on moduli space}. Surveys in Differential Geometry~{\bf 1} (1991), 243--310.

\end{thebibliography}
\end{document}